\documentclass[12pt]{article}
\usepackage{}
\usepackage{CJK}
\usepackage[round]{natbib}
\usepackage{mathrsfs}
\usepackage{caption}
\usepackage{amsfonts}
\usepackage{tabularray}
\usepackage{float}
\usepackage{booktabs}
\usepackage{amsmath}
\usepackage{amssymb}
\usepackage{dsfont,footmisc}
\usepackage[all]{xy}
\usepackage{graphicx}
\usepackage{subfigure}
\usepackage{color}
\usepackage{subcaption}
\usepackage[singlespacing]{setspace}
\usepackage{multirow}
\usepackage{rotating}
\usepackage{comment}
\usepackage{tabularx}
\usepackage[font=small]{caption}%
\usepackage[colorlinks,linkcolor=red,anchorcolor=blue,citecolor=blue]%
{hyperref}
\usepackage{booktabs}    
\usepackage{multirow}    
\usepackage{threeparttable} 
\usepackage{array}
\usepackage{verbatim}
\def\d{\mathrm{d}}

\newcommand{\VaR}{\mathrm{VaR}}

\newcommand{\RV}{\mathcal{RV}}

\newcommand{\SES}{\mathrm{SES}}

\newcommand{\MES}{\mathrm{MES}}

\newcommand{\E}{\mathbb{E}}
\newcommand{\R}{\mathbb{R}}

\newcommand{\p}{\mathbb{P}}

\newcommand{\id}{\mathds{1}}

\renewcommand{\)}{\right)}

\newtheorem{theorem}{Theorem}[section]

\newtheorem{lemma}{Lemma}[section]

\newenvironment{proof}[1][Proof]{\noindent \textbf{#1.} }{\  \rule{0.5em}{0.5em}}
\newtheorem{proposition}{Proposition}[section]

\def\nn{\nonumber}

\begin{document}
\baselineskip=16pt
\title{Asymptotics of Systemic Risk in a Renewal Model with Multiple Business Lines and Heterogeneous Claims}

\author{ 
    Bingzhen Geng\thanks{\scriptsize  School of Big Data and Statistics, Anhui University, Hefei, Anhui 230601, China. Email: \texttt{gengbz@ahu.edu.cn}}
	\and
	Yang Liu\thanks{\scriptsize Corresponding Author. School of Science and Engineering, The Chinese University of Hong Kong, Shenzhen, Shenzhen, Guangdong 518172, China. Email: \texttt{yangliu16@cuhk.edu.cn}} 
	\and
	Hongfu Wan\thanks{\scriptsize School of Science and Engineering, The Chinese University of Hong Kong, Shenzhen, Shenzhen, Guangdong 518172, China. Email: \texttt{120090003@link.cuhk.edu.cn}}
}

\date{}
\maketitle
\begin{abstract}	
	
	
   Systemic risk is receiving increasing attention in the insurance industry. In this paper, we propose a multi-dimensional L\'{e}vy process-based renewal risk model with heterogeneous insurance claims, where every dimension indicates a business line of an insurer. We use the systemic expected shortfall (SES) and marginal expected shortfall (MES) defined with a Value-at-Risk (VaR) target level as the measurement of systemic risk. Assuming that all the claim sizes are pairwise asymptotically independent (PAI), we derive asymptotic formulas for the tail probabilities of discounted aggregate claims and the total loss, which hold uniformly for all time horizons. We further obtain the asymptotics of the above systemic risk measures. The main technical issues involve the treatment of uniform convergence in the dynamic time setting. Finally, we perform a detailed Monte Carlo study to validate our asymptotics and analyze the impact and sensitivity of key parameters in the asymptotic expressions both analytically and numerically.
 
    \quad \\[1pt]
\end{abstract}

\section{Introduction}\label{sec:1}




 Systemic risk is becoming increasingly crucial in the insurance industry. A large insurance company often operates across multiple business lines associated with a variety of industries, and thus is often closely intertwined with financial markets through complex reinsurance networks, partnerships with private equity, and other investment arrangements (\cite{labini2024systemic}).  Additionally, extreme events such as earthquakes or wars can trigger intense short-term claims across multiple business lines, leading to a huge risk exposure if their capital is not properly allocated. 
Coupled with investment losses due to financial volatility, the resulting capital shortage in large insurance institutions can pose systemic risk to the whole economy; e.g., American International Group (AIG) in the 2008 global financial crisis. 

In recent years, both scholars and practitioners have developed various tools to measure systemic {risk} from different perspectives, including connectivity measures, 
illiquidity measures, 
probability distribution measures and so on (\cite{pampurini2024measuring}). 
Some tools related to probability distribution measures
include capital shortfall in \cite{acharya2012capital}, CoVaR in \cite{tobias2016covar}, scenario-based risk measures in \cite{wang2021scenario}, robust risk aggregation approaches in \cite{chen2022ordering} and \cite{BLLW25}, generalized risk measures in \cite{fadina2024framework}, and VaR-based and expectile-based systemic risk measures in \cite{geng2024value}. In this paper, we are going to investigate two specific probability distribution measures, systemic expected shortfall ($\SES$) and marginal expected shortfall ($\MES$), from an asymptotic aspect to study the capital shortage risk of an insurer 
given a systemic extreme scenario.

{
Specifically, $\SES$ and $\MES$ target at measuring an economic agent's 
capital shortfall when the whole system is undercapitalized (see \cite{acharya2017measuring} and \cite{chen2022asymptotic}). The feature of SES is that the risk faced by one agent is quantified by the expectation of the exceeding loss variable of an individual target level, conditional on the catastrophic event that the system loss exceeds a total target level; MES also has a similar feature. They are defined in various ways in the literature because different target levels are adopted. In this paper, we 
use the approach of setting both the individual and total target levels as Value-at-Risk (VaR) in \cite{asimit2011asymptotics} and \cite{jaune2022asymptotic}, and apply it to our definition of SES and MES.} Formally, for $1\le k\le d$, let $Z_{k}$ be the individual risk of the $k$-th economic agent, and let $S:=\sum_{i=1}^{d}Z_{i}$ be the aggregate risk. We have 
\begin{align*}
&\SES_{q,k}(S):=\E[\left(Z_{k}-\VaR_{q}\left(Z_{k}\right)\right)^{+}\mid S>\VaR_{q}\left(S\right)],\\
&\MES_{q,k}(S):=\E[Z_{k}\mid S>\VaR_{q}\left(S\right)],   
\end{align*}
where $$\VaR_{q}(X):=F_{X}^{\gets}(q)=\inf \{y\in\mathbb{R}:F_{X}(y)\ge q\},\, 0<q<1$$ is the quantile function of a random variable $X$ with distribution $F_{X}$. 

{In this paper, we construct a dynamic model, specifying the individual and aggregate {risk} by a $d$-dimensional L\'{e}vy process-based renewal model with heterogeneous claims. Consider a large insurance {company} consisting of $d$ business lines. The {risk} in each business line consists of insurance risk and financial risk, where the former arises from monetary losses due to insurance claims, and the reserve of premium income after deducting claim payouts is invested in the financial market, whose fluctuations then bring financial risk. We use $Z_{k,t}$ to denote the discounted loss in the time span $[0,t]$ of the $k$-th business line, expressed by a renewal {risk} model where the discount factor is a geometric L\'{e}vy process. Then, $D_t:=\sum_{k=1}^dZ_{k,t}$ is the discounted total loss. We assume that all the claim sizes are pairwise asymptotically independent (PAI), and are independent of all inter-arrival times. Additionally, we assume every business line has multiple claim types, as a single business line typically involves not only one type of claim, and the occurrence of one claim is often inevitably accompanied by a series of other heterogeneous claims. {We also assume the arrival time processes of these heterogeneous claims in each business line are arbitrarily dependent, so as to depict the temporal dependence of these claims.} For example, in a business line of traffic accident insurance, a car accident may lead to a property damage claim that subsequently generates a medical claim. Under capital shortage, such {temporal dependence} can increase the pressure on solvency.}

Then, as we mentioned above, $\SES$ and $\MES$ can be employed to assess the individual risk (particularly, the insurer's loss in a particular business line) under the condition of a systemic crisis. More precisely, we study systemic risk of each business line via
$$\SES_{q,k}(D_{t})=\E[\left(Z_{k,t}-\VaR_{q}\left(Z_{k,t}\right)\right)^{+}\mid D_{t}>\VaR_{q}\left(D_{t}\right)],$$
$$\MES_{q,k}(D_{t})=\E[Z_{k,t}\mid D_{t}>\VaR_{q}\left(D_{t}\right)],\;\, 1\le k\le d,\, 0<q<1,\, t\ge 0.$$
{It is important to note that here systemic risk focuses on the degree of capital shortage of a large insurance {company} under extreme tail events. We aim to provide asymptotic approximation formulas (namely, as $q$ approaches $1$) of $\SES_{q,k}(D_{t})$ and $\MES_{q,k}(D_{t})$ that can be directly calculated, to assist the manager to conduct appropriate risk assessment and capital pre-allocation.}

Actually, the renewal { risk} model has been studied in different aspects. For example, \cite{tang2007heavy} obtains the explicit asymptotic expression of the tail probability of discounted losses for a one-dimensional case, or more precisely, when the company has only one business line. Then, \cite{tang2010uniform} extend this result to a L\'{e}vy process-based case. \cite{li2012asymptotics} considers a time-dependent one-dimensional case. \cite{li2022asymptotic} considers a multi-dimensional case. As for the case that one business line contains two different kinds of claims, there are also many explorations. To the best of our knowledge, it is first taken into account by \cite{yuen2005ultimate}. Then in several works, such as \cite{li2013pairwise} and \cite{yang2019asymptotic}, the one-dimensional case has been discussed and they call this case ``by-claim" or ``delayed claim". Also note that most of the above literature focuses on the ruin probabilities of their renewal models, while \cite{li2022asymptotic} goes further to quantify a corresponding systemic risk measure.

Here, our paper considers a more general multi-dimensional L\'{e}vy process-based renewal model with multiple heterogeneous claims. Then, we quantify systemic risk of each business line with the assistance of two risk measures $\SES_{q,k}(D_{t})$ and $\MES_{q,k}(D_{t})$. 
{The main contributions and summary of this paper are as follows. 

(a) We develop a multi-dimensional renewal {risk} model to accommodate a more general setting involving multiple claim types per business line. We provide asymptotic expressions of the tail probability of the discounted total loss and  
those of the {$\SES_{q,k}(D_{t})$ and $\MES_{q,k}(D_{t})$} in the renewal model, which can be directly computed in practice. While \cite{li2022asymptotic} applies a systemic risk measure with linear target levels to assess the renewal model, our approach 
applies 
nonlinear $\VaR$ target levels to more accurately evaluate the risk.  

(b) We provide a technical extension on the issue of uniform convergence in the asymptotic analysis of {$\SES_{q,k}(D_{t})$ and $\MES_{q,k}(D_{t})$}. Unlike static models focusing on random variables with pre-assumed dependence structures (see, e.g., \cite{chen2022asymptotic}), our paper incorporates a dynamic setting into systemic risk measures. This requires further consideration on the uniformity of asymptotic properties over the whole time scale, which will be discussed in Section \ref{sec:5} in detail.  

(c) We conduct an in-depth analysis of the validity of the asymptotic approximations and the impact and sensitivity of parameter variations, and conclude that higher fluctuation parameters can lead to larger asymptotic values of SES and MES.
}

The rest of this paper is structured as follows. Section \ref{sec:2} provides some preliminaries on the notation and basic notions of regularly varying distributions, L\'{e}vy processes, and asymptotic independence. Section \ref{sec:3} presents the model setup and our main results. Section \ref{sec:4} conducts numerical studies to validate the accuracy of our results and analyze the impact and sensitivity of different parameters. Section \ref{sec:5} collects all the proofs.

\section{Preliminaries}\label{sec:2}
Throughout this paper, all limiting formulas either contain $x$ or contain $q$, and these two cases are mutually exclusive. In the following, all limit relationships refer to $x\rightarrow\infty$ if the formulas contain $x$, and refer to $q\uparrow 1$ if the formulas contain $q$, unless stated otherwise. For two positive functions $f(\cdot)$ and $g(\cdot)$, write $f(x)=O(g(x))$ if $\limsup f(x)/g(x)<\infty$; write $f(x)=o(g(x))$ if $\lim f(x)/g(x)=0$; write $f(x)\lesssim g(x)$ if $\limsup{f(x)}/{g(x)}\leq 1$; write $f(x)\gtrsim g(x)$ if $\liminf{f(x)}/{g(x)}\geq 1$; write $f(x)\sim g(x)$ if $\lim{f(x)}/{g(x)}=1$; we say $f$ and $g$ are weakly equivalent, denoted by $f(x)\asymp g(x)$, if both $f(x)=O(g(x))$ and $g(x)=O(f(x))$. Furthermore, for two positive bivariate functions $f(\cdot,\cdot)$ and $g(\cdot,\cdot)$, we say that the asymptotic relation $f(x,t)\sim g(x,t)$ holds uniformly for all $t$ in a nonempty set $\Delta$ if $$ \lim_{x\to \infty}\sup_{t\in\Delta} \Big|\frac{f(x,t)}{g(x,t)}-1\Big|=0.$$ Also, we say $f(x,t)\lesssim g(x,t)$ holds uniformly for all $t\in\Delta$ if $$\limsup_{x\to\infty}\sup_{t\in\Delta}\frac{f(x,t)}{g(x,t)}\le 1.$$ It is easy to see that $f(x,t)\sim g(x,t)$ uniformly for all $t\in\Delta$ if $f(x,t)\lesssim g(x,t)$ and $g(x,t)\lesssim f(x,t)$ both hold uniformly for all $t\in \Delta$. As usual, for a random variable $X$, write $X^+=\max\{X,0\}$. The indicator function of an event $A$ is denoted by $\id_{A}$. For any distribution function $F$, denote its tail by $\overline F(x)=1-F(x)$. For two real numbers $a$ and $b$, write $a\vee b=\max\{a,b\}$, and $a\wedge b=\min\{a,b\}$.

In this paper, a specific class of distributions is concerned. We say a distribution function $F$ has a regularly varying tail, if there is a fixed constant $0<\alpha<\infty$ such that for all $y>0$, $$\lim_{x\to\infty}\frac{\overline{F}(yx)}{\overline{F}(x)}= y^{-\alpha},$$ which is further denoted by $\overline{F}\in \RV_{-\alpha}$. By Theorem 1.5.6 of \cite{bingham1989regular}, if $\overline{F}\in \RV_{-\alpha}$, then for all $\varepsilon>0$ {and} all $b>1$, there is some $x_{0}>0$ such that Potter's bounds
\begin{align*}
\frac{1}{b}\left(\left(\frac{y}{x}\right)^{-\alpha-\varepsilon}\wedge\left(\frac{y}{x}\right)^{-\alpha+\varepsilon}\right)\le\frac{\overline{F}(y)}{\overline{F}(x)}\le b\left(\left(\frac{y}{x}\right)^{-\alpha-\varepsilon}\vee\left(\frac{y}{x}\right)^{-\alpha+\varepsilon}\right)    
\end{align*}
hold for any $x,y\ge x_{0}$. This implies for any $\beta>\alpha$,
\begin{align}
x^{-\beta}=o(\overline{F}(x)).\label{pre1}
\end{align}

In this paper, the L\'{e}vy process $\{R_{t}\}_{t\geq 0}$ is assumed to be right-continuous with a left limit, satisfying $\E\left[R_{1}\right]>0$. Then $\E\left[R_{t}\right]$ tends to $\infty$ as $t\to\infty$. The Laplace exponent is defined by 
\begin{align}
\phi(\alpha):=\log \E \left[e^{-\alpha R_{1}}\right]\label{pre3}
\end{align}
for $\alpha\in\mathbb{R}$. If $\phi(\alpha)$ is finite, then $\E\left[e^{-\alpha R_{t}}\right]=e^{t\phi(\alpha)}<\infty$; see \cite{tang2010uniform} for further acquaintance with the L\'{e}vy process and Laplace exponent.

Two non-negative and unbounded random variables $Z_{1}$ and $Z_{2}$ with distributions $F_{1}$ and $F_{2}$ are said to be asymptotically independent (abbreviated as AI) if $$\lim_{x\to\infty}\frac{\p\left(Z_{1}>U_{1}(x), Z_{2}>U_{2}(x)\right)}{\p\left(Z_{1}>U_{1}(x)\right)}=0$$ where $U_{i}(x):=\left(1/\overline{F_{i}}\right)^{\gets}(x),i=1,2$. According to Proposition 2.1 of \cite{li2022asymptotic}, if we have further $\overline{F_{i}}\in \RV_{-\alpha},i=1,2$, and $\overline{F_{1}}(x)\asymp\overline{F_{2}}(x)$, the asymptotic independence of $Z_{1}$ and $Z_{2}$ is equivalent to 
\begin{align}
\lim_{x\to\infty}\frac{\p\left(Z_{1}>x,Z_{2}>x\right)}{\p\left(Z_{1}>x\right)}=0.\label{pre2}  
\end{align}
Note that the asymptotic independence is one of practical dependence structures. For instance, two independent random variables are naturally asymptotically independent, and it can be depicted by certain copulas such as {Ali-Mikhail-Haq copula, Farlie-Gumbel-Morgenstern copula, Frank copula and so on; see \cite{li2010subexponential} and \cite{mcneil2015quantitative}.} Some useful general asymptotic properties of AI are discussed in \cite{chen2009sums}. Further, a series of random variables $Z_{1}, Z_{2},\ldots$ are said to be pairwise asymptotically independent (PAI) if each two of them are asymptotically independent.

\allowdisplaybreaks[1]
\section{Model setup and main results}\label{sec:3}

\subsection{Model setup}\label{sec:3.1}
{ Recall the setting explained in Section \ref{sec:1}. Especially, for the $k$-th business line ($1\leq k\leq d$), we define the surplus process (or insurance risk process as stated in \cite{kluppelberg2008integrated}) as a variation of the traditional Cram\'{e}r–Lundberg model:
\begin{align*}
U_{k,t}:=u_k+c_kt-\Psi_{k,t}:=u_k+c_kt-\sum_{j=1}^{r}\sum_{i=1}^{N_{k,t}^j}X_{ki}^j,
\end{align*}
where $u_k>0$ is the initial capital, $c_k>0$ is the premium rate and $\Psi_{k,t}$ is the claim payment process. Within $\Psi_{k,t}$, $r$ is the number of types of heterogeneous claims ($1\leq r< \infty$). 
The nonnegative random variables $\{X_{ki}^j;i\geq 1\}$ denote the sizes of the $j$-th  type of claims in the $k$-th business line and have a common distribution $F_k^j$. Let $\{\theta_{ki}^j;i\geq 1\}$ be the $j$-th type of claims' inter-arrival times from the $k$-th business line, and denote by $\tau_{ki}^j=\sum_{h=1}^i\theta_{kh}^j$ and 
$N_{k,t}^j=\sup\{m\geq 0:\tau_{km}^j\le t\}$  the corresponding claim
arrival times and counting processes. Assume $\tau_{k0}^j=0$ for convenience in notation. The insurer's discount factor is expressed as a geometric L\'{e}vy process $\{e^{-R_{t}}\}_{t\ge 0}$, where the investment return $\{R_{t}\}_{t\ge 0}$ is a L\'{e}vy process. In order to assess the risk generated by the $k$-th business line over a given period $[0,t]$ at its beginning, we consider its discounted net loss process (as named in \cite{kluppelberg2008integrated}):
\begin{align*}
Z_{k,t}:={-\int\limits_{0-}^te^{-R_{s-}}\mathrm{d}U_{k,s}}=&{-\int\limits_{0-}^te^{-R_{s-}}(c_k\mathrm{d}s-\mathrm{d}\Psi_{k,s})}\\
=&\sum_{j=1}^{r}\sum_{i=1}^{N_{k,t}^j} X_{ki}^je^{-R_{\tau_{ki}^j}} -c_{k}\int\limits_{0}^{t}e^{-R_{s}}\mathrm{d} s.
\end{align*}
Hence, the discounted net loss process of the whole insurance company is:
\begin{align*}
D_t:=\sum_{k=1}^{d}Z_{k,t}=&\sum_{k=1}^{d}\sum_{j=1}^{r} \sum_{i=1}^{N_{k,t}^j}X_{ki}^je^{-R_{\tau_{ki}^j}}-\sum_{k=1}^{d}c_{k}\int\limits_{0}^{t}e^{-R_{s}}\mathrm{d} s\\
=&\sum_{k=1}^d\sum_{j=1}^{r}\sum_{i=1}^{\infty}X_{ki}^je^{-R_{\tau_{ki}^j}}\id_{\left \{ \tau_{ki}^j\le t \right \}} -\sum_{k=1}^{d}c_{k}\int\limits_{0}^{t}e^{-R_{s}}\mathrm{d} s\\
=:&S_{t}-\sum_{k=1}^{d}c_{k}\int\limits_{0}^{t}e^{-R_{s}}\mathrm{d} s.
\end{align*}
Here $S_t$ can be seen as the discounted present value of all the $d$ business lines' aggregate claims. 

For the dependence structure, assume that all claim sizes $\{X_{ki}^j; i\geq 1,\,1\le k\le d, \,1\leq j\leq r\}$ are 
PAI (
defined in Section \ref{sec:2}). Meanwhile, for any $1\le k\le d$, suppose that  
$\{\left(\theta_{ki}^1,\theta_{ki}^2,\ldots,\theta_{ki}^{r}\right);i\geq 1\}$ is a sequence of independent and identically distributed (i.i.d.) nonnegative random vectors (and thus all $N_{k,t}^j$ are renewal processes), but 
the dependence structure among all the components in each vector can be arbitrary.
{ This setup is to accommodate an adequate setting which can capture the inherent temporal dependence of these heterogeneous claims in a certain business line. As the sequential occurrence of property damage and medical claims in the aforementioned car accident example, such temporal dependence 
can make the joint occurrence of heterogeneous claims more likely, thereby potentially inducing greater losses in that line and consequently contributing to systemic risk. 
}Finally, all remaining possible combinations of random variables, such as $\{X_{ki}^j;i\geq 1,\,1\leq j\leq r,\,1\le k\le d\}$, $\{\(N_{1,t}^{1},N_{1,t}^{2},\ldots,N_{1,t}^{r}\),\ldots,\(N_{d,t}^{1},N_{d,t}^{2},\ldots,N_{d,t}^{r}\);t\geq 0\}$, and $\{e^{-R_{t}}\}_{t\ge 0}$, are assumed to be mutually independent. 

For $1\le k\le d,\,1\le j\le r$, define $$\lambda_{k,t}^{j}:=\E\left[N_{k,t}^{j}\right]=\sum_{i=1}^{\infty}\p(\tau_{ki}^j\le t).$$ This paper mainly considers the asymptotic properties on the set $$\Lambda:=\{t\ge 0:0<\lambda_{k,t}^{j}\le\infty, 1\le k\le d,\,1\le j\le r\}.$$ Namely, if $\underline{t}:=\inf\{t\ge 0:\underset{\substack{1\le k\le d,\\1\le j\le r}}{\min}\p(\tau_{k1}^j\le t)>0\},$ then 
\begin{align*}
\Lambda=\begin{cases}
[\underline{t},\infty],\; \text{if}\; \underset{\substack{1\le k\le d,\\1\le j\le r}}{\min}\p(\tau_{k1}^j\le \underline{t})>0;\\
(\underline{t},\infty],\; \text{if}\; \underset{\substack{1\le k\le d,\\1\le j\le r}}{\min}\p(\tau_{k1}^j\le \underline{t})=0.
\end{cases}
\end{align*}
For notational convenience, for $T\in\Lambda$, we use $\Lambda_{T}$ to denote $\Lambda\cap [0,T]$, and $\Lambda^{T}$ to denote $\Lambda\cap [T,\infty]$.

To sum up, in the above model, systemic risk primarily arises from three sources: (a) the joint effect of claims of $d$ business lines, (b) the temporal dependencies among heterogeneous claims within one business line (e.g., when a first claim is necessarily followed by a second), and (c) the financial market fluctuations affecting the investment return of the entire system. 

}

\subsection{Main results}\label{sec:3.2}

We formally present our first result. We give asymptotic expressions of the tail probabilities of the discounted present value of the aggregate claims $S_{t}$ and that of the total loss $D_{t}$ under the proposed model above.
\begin{theorem}\label{thm1}
Consider the $d$-dimensional renewal model with heterogeneous claims introduced in Section \ref{sec:3.1}. Assume for all $1\le k\le d,\,1\le j\le r$, $\overline{F_{k}^j}\in \RV_{-\alpha}$ for some $\alpha>0$, and for all $1\le k_1, k_2\leq d,1\leq j_1,j_2\leq r$, 
$\overline{F_{k_1}^{j_1}}(x)\asymp \overline{F_{k_2}^{j_2}}(x)$.
For the Laplace exponent of the L\'{e}vy process $R_{t}$, suppose there is some $\alpha^{\ast}>\alpha$ such that $\phi(\alpha^{\ast})<0$.  As $x\rightarrow\infty$, we have:\\
({\romannumeral1}) it holds uniformly for all $t\in\Lambda$ that
\begin{align}
\p\left(S_{t}>x\right)\sim \sum_{k=1}^{d}\sum_{j=1}^{r} \overline{F_{k}^j}\left ( x \right )\int\limits_{0-}^{t}e^{s\phi \left ( \alpha \right ) }\mathrm{d}\lambda_{k,s}^{j} ;\label{thm1_1}
\end{align}
({\romannumeral2}) for any fixed $T\in\Lambda$, it holds uniformly for all $t\in\Lambda^{T}$ that
\begin{align}
\p\left(D_{t}>x\right)\sim \sum_{k=1}^{d}\sum_{j=1}^{r} \overline{F_{k}^j}\left ( x \right )\int\limits_{0-}^{t}e^{s\phi \left ( \alpha \right ) }\mathrm{d}\lambda_{k,s}^{j} .\label{thm1_2}   
\end{align}

\end{theorem}


By \eqref{pre3}, $\phi(\cdot)$ is convex for which $\phi(\cdot)$ is finite. Indeed, by H\"{o}lder's inequality, for any $0<\upsilon<1$, for any $a,b$ such that $\phi(a)$ and $\phi(b)$ are finite, we have
{
\begin{align*}
\phi(\upsilon a +(1-\upsilon)b)&=\log \E\left[e^{-(\upsilon a+(1-\upsilon)b) R_{1}}\right]\\&\le\log\left[\left(\E\left[e^{-a R_{1}}\right]\right)^{\upsilon}\left(\E\left[e^{-b R_{1}}\right]\right)^{1-\upsilon}\right]\\
&=\upsilon\log\E\left[e^{-a R_{1}}\right]+(1-\upsilon)\log\E\left[e^{-b R_{1}}\right]\\
&=\upsilon\phi(a)+(1-\upsilon)\phi(b).  
\end{align*}}
Therefore, since $\phi(0)=0$ and $\phi(\alpha^{\ast})<0$, we have for any {$\kappa\in (0,\alpha^{\ast}]$, $\phi(\kappa)<0$}, and thus $\phi(\alpha)<0$ as well. Conversely, since $\phi$ is continuous, if $\phi(\alpha)<0$, there must exist some {$\alpha^{\ast}>\alpha$} such that $\phi(\alpha^{\ast})<0$. Thus, the requirement that there is some $\alpha^{\ast}>\alpha$ such that $\phi(\alpha^{\ast})<0$ is actually equivalent to $\phi(\alpha)<0$.  

Note that the ``$0-$" in the integral sign means approaching 0 from the left side of 0. It is set to avoid that the right-hand side of the asymptotic relation in \eqref{thm1_1} equals to 0 when $t=0$ (if $0\in\Lambda$). Also note that for any $t\in\Lambda$, for all $1\le k\le d,\,1\le j\le r$, we have $$0<\int\limits_{0-}^{t}e^{s\phi \left ( \alpha \right ) }\p\left(\tau_{ki}^j\in\mathrm{d}s\right)<\p\left(\tau_{ki}^j\le t\right)\le 1.$$ Hence, 
\begin{align}
0<\int\limits_{0-}^{t}e^{s\phi \left ( \alpha \right ) }\mathrm{d}\lambda _{k,s}^{j}\le\sum_{i=1}^{\infty}\E\left[e^{\tau_{ki}^j\phi(\alpha)}\right]&=\sum_{i=1}^{\infty}\E\left[e^{\sum_{p=1}^{i}\left(\tau_{kp}^j-\tau_{k(p-1)}^j\right)\phi(\alpha)}\right] \nonumber\\ 
&=\sum_{i=1}^{\infty}\left(\E\left[e^{\tau_{k1}^j\phi(\alpha)}\right]\right)^{i}=\frac{\E\left[e^{\tau_{k1}^j\phi(\alpha)}\right]}{1-\E\left[e^{\tau_{k1}^j\phi(\alpha)}\right]}<\infty. \label{integral<inf}
\end{align}
These conclusions will be very useful in the proof of Theorem \ref{thm1}.

From the right-hand side of the asymptotic expression \eqref{thm1_2}, we can see the ingredients that involve $t$ and those that involve claim sizes (i.e., $\overline{F_{k}^j}(x)$) are separated, and the latter ones indicate the decay rates, while the former ones act as the ``weight" of each corresponding claim.

Now, we formally present our second main result. We give asymptotic expressions of { $\SES_{q,k}(D_{t})$ and $\MES_{q,k}(D_{t})$} of each business line under the proposed model. 
\begin{theorem}\label{thm2}
Under the {setting} in Theorem \ref{thm1}, suppose further there is a distribution function $F$ such that for all $ 1\le k\le d,\,1\le j\le r, \overline{F_{k}^j}(x) \sim a_{k}^j\overline{F} \left ( x \right )$ for some $a_{k}^j>0$. Further, we assume $\alpha>1$. Then for any fixed $T\in\Lambda$, as $q\uparrow 1$, it holds uniformly for all $t\in\Lambda^{T}$ that 
\begin{align}
&\SES_{q,k}(D_{t}) \sim \frac{l_{k}(t)}{\sum_{i=1}^{d}l_{i}(t)} \left ( \left(\sum_{i=1}^{d}l_{i}(t)\right)^{\frac{1}{\alpha}}-\left ( l_{k}(t) \right )^{\frac{1}{\alpha}} + \frac{\left ( \sum_{i=1}^{d}l_{i}(t) \right )^{\frac{1}{\alpha} } }{\alpha-1}\right ) F^{\gets }\left ( q \right),\label{ses}\\
&\MES_{q,k}(D_{t}) \sim \frac{\alpha }{\alpha-1} \frac{l_{k}(t)}{\left( \sum_{i=1}^{d}l_{i}(t) \right )^{1-\frac{1}{\alpha}}} F^{\gets }\left ( q \right ), \label{mes}
\end{align}
where $l_{k}(t):= \sum_{j=1}^{r}a_{k}^j\int\limits_{0-}^{t}e^{s\phi(\alpha)}\mathrm{d}\lambda_{k,s}^{j},\: 1 \le k \le d.$
\end{theorem}

From the right-hand sides of \eqref{ses} and \eqref{mes}, it is easy to see that the quantile function of $F$
is the main asymptotic term of $\SES_{q,k}(D_{t})$ and $\MES_{q,k}(D_{t})$, while other information such as, each business line's claim-size ratios to $F$ (i.e., $a_{k}^j$), financial {risk}, and claim frequency, only contribute to the coefficients of the formulas. To further explore properties of these coefficients, fix some $t\in\Lambda^{T}$ and  $\sum_{i=1}^{d}l_{i}(t)$. For a certain $1\le k\le d$, define $\rho_{k}:=\frac{l_{k}(t)}{\sum_{i=1}^{d}l_{i}(t)}$. In other words, $0<\rho_{k}<1$ indicates the proportion of the combined information of the $k$-th business line's claim-size ratios, financial risk and claim frequency in that combined information of the entire entity. Further, define $h_{\SES}(\rho_{k}):=\rho_{k}\left(1-\rho_{k}^{\frac{1}{\alpha}}+\frac{1}{\alpha-1}\right),\,h_{\MES}(\rho_{k}):=\frac{\alpha}{\alpha-1}\rho_{k}$. Then \eqref{ses} and \eqref{mes} can be rewritten as:
\begin{align*}
&\SES_{q,k}(D_{t})\sim\rho_{k}\left(1-\rho_{k}^{\frac{1}{\alpha}}+\frac{1}{\alpha-1}\right)\left(\sum_{i=1}^{d}l_{i}(t)\right)^{\frac{1}{\alpha}}F^{\gets }\left ( q \right )=h_{\SES}(\rho_{k})\times\left(\sum_{i=1}^{d}l_{i}(t)\right)^{\frac{1}{\alpha}}F^{\gets }\left ( q \right ),\\
&\MES_{q,k}(D_{t})\sim\frac{\alpha}{\alpha-1}\rho_{k}\left(\sum_{i=1}^{d}l_{i}(t)\right)^{\frac{1}{\alpha}}F^{\gets }\left ( q \right )=h_{\MES}(\rho_{k})\times\left(\sum_{i=1}^{d}l_{i}(t)\right)^{\frac{1}{\alpha}}F^{\gets }\left ( q \right ).
\end{align*}
Note that 
\begin{align*}
\frac{\mathrm{d}h_{\text{SES}}(\rho_{k})}{\mathrm{d}\rho_{k}}=-\frac{\alpha+1}{\alpha}\rho_{k}^{\frac{1}{\alpha}}+\frac{\alpha}{\alpha-1},\,\,\, \frac{\mathrm{d}h_{\text{MES}}(\rho_{k})}{\mathrm{d}\rho_{k}}=\frac{\alpha}{\alpha-1}.
\end{align*}
Obviously, $\frac{\mathrm{d}h_{\text{MES}}(\rho_{k})}{\mathrm{d}\rho_{k}}>0$ for all $0<\rho_{k}<1$; $\frac{\mathrm{d}h_{\SES}(\rho_{k})}{\mathrm{d}\rho_{k}}$ is decreasing in $\rho_{k}$, and $\frac{\mathrm{d}h_{\text{SES}}(\rho_{k})}{\mathrm{d}\rho_{k}}\big|_{\rho_{k}=0}=\frac{\alpha}{\alpha-1}>0,\, \frac{\mathrm{d}h_{\text{SES}}(\rho_{k})}{\mathrm{d}\rho_{k}}\big|_{\rho_{k}=1}=\frac{1}{\alpha(\alpha-1)}>0$. Hence, $\frac{\mathrm{d}h_{\text{SES}}(\rho_{k})}{\mathrm{d}\rho_{k}}>0$ for all $0<\rho_{k}<1$. Therefore, we can see both the coefficients of asymptotic expressions of $\MES_{q,k}(D_{t})$ and $\SES_{q,k}(D_{t})$ increase with respect to $\rho_{k}$. This conclusion aligns with anticipations: as the $k$-th business line's proportion of claim-size ratios, financial {risk}, and claim frequency within the entire entity increases, the asymptotic expressions for $\SES_{q,k}(D_{t})$ and $\MES_{q,k}(D_{t})$ also increase, indicating a corresponding increase in the expected systemic risk for that business line. {The impact of parameters related to financial {market fluctuations}, claim frequency and other factors will be further discussed in the next section.}

\section{Numerical study}\label{sec:4}
{ In this section, we employ numerical methods to verify the accuracy of our main results and analyze the impact and sensitivity of variations of each parameter on the asymptotic approximations. In all the following subsections, we use $p^A(x,t)$, $\SES^A(q,k,t)$ and $\MES^A(q,k,t)$ to denote the asymptotic approximations of $\p(D_t>x)$, $\SES_{q,k}(D_t)$ and $\MES_{q,k}(D_t)$, and use $p^E(x,t)$, $\SES^E(q,k,t)$ and $\MES^E(q,k,t)$ to denote the corresponding empirical estimates.

\subsection{Model instantiation and accuracy examination}\label{subsection4.1}
To specify the setting, we assume the insurer has two business lines (i.e., $d=2$), each of which has two heterogeneous claims (i.e., $r=2$). For $k=1,2,\,j=1,2$, suppose that their corresponding claim sizes $\{X_{ki}^j;i\geq 1\}$ follow respectively Pareto distributions with a common parameter $\alpha>1$ and four different parameters $\gamma_{k}^j$, namely, 
$$
F_{k}^j(x)=1-\left(\frac{\gamma_{k}^j}{\gamma_{k}^j+x}\right)^{\alpha}, \; x\ge 0.
$$ 
Clearly, $\overline{F_{k}^j}\in \RV_{-\alpha}$, and $\overline{F_{k}^j}(x)\sim\left(\frac{\gamma_{k}^j}{\gamma_{1}^1}\right)^{\alpha}\overline{F_{1}^1}(x)$. The dependence structure of these claim sizes is characterized by a Frank copula $$C(u)=-\frac{1}{\Theta}\mathrm{In}\left(1-\left(1-e^{-\Theta}\right)^{p+1}\prod_{i=1}^{p}\left(1-e^{-\Theta u_{i}}\right)\right),u_{i}\in [0,1]\: \mbox{for}\, i=1,\ldots,p,\,\text{where }p\in\mathbb{N}^+,
$$
so that they are guaranteed to be PAI. 

In the meantime, we specify the investment return process $\{R_t\}_{t\geq0}$ to be
\begin{align}
R_t=\mu t+\sigma W_t+\sum_{i=1}^{N_t^{\prime}}\zeta_i,  \label{Rt}  
\end{align}
where $\mu,\sigma\ge0$ are parameters of the drift and diffusion terms, $\{W_t\}_{t\geq0}$ is a standard Brownian motion, and $\sum_{i=1}^{N_t^{\prime}}\zeta_i$ is a compound Poisson process with i.i.d. standard normal random variables $\zeta_i, i=1,2,\dots$, and $\{N_t^{\prime}\}_{t\geq0}$ is a homogeneous Poisson process with intensity $\lambda$. All terms in $\{R_{t}\}_{t\ge 0}$ are assumed to be mutually independent and independent of other components in the model. Hence, the Laplace exponent of $\{R_{t}\}_{t\ge 0}$ is
\begin{align}
\phi_R(\alpha)=-\mu\alpha+\frac{1}{2}\sigma^2\alpha^2+\lambda\left(\E\left[e^{-\alpha \zeta_1}\right]-1\right)=-\mu\alpha+\frac{1}{2}\sigma^2\alpha^2+\lambda\left(e^{\frac{\alpha^2}{2}}-1\right). \label{phi_Rt}
\end{align}

Also, specify all the four renewal processes $\{N_{k,t}^{j}\}_{t\geq0}$ to be {independent} homogeneous Poisson processes with corresponding intensity parameters $\lambda_{k}^j$. 
Here, we focus on independent renewal processes for computational tractability and notational simplicity. Although we omit dependent-arrival simulations, our theoretical results accommodate arbitrary temporal dependence and establish the same asymptotic conclusions.  
Therefore, according to Theorems \ref{thm1} and \ref{thm2}, by simple calculations, we have that the asymptotic approximations are: 
\begin{align*}
&p^A(x,t)=\sum_{k=1}^2\sum_{j=1}^2\lambda_k^j\left(\frac{\gamma_k^j}{\gamma_k^j+x}\right)^{\alpha}\frac{e^{t\phi_R(\alpha)}-1}{\phi_R(\alpha)},\\
&\SES^A(q,k,t)=\frac{l_{k}(t)}{\sum_{i=1}^{2}l_{i}(t)} \left ( \left(\sum_{i=1}^{2}l_{i}(t)\right)^{\frac{1}{\alpha}}-\left ( l_{k}(t) \right )^{\frac{1}{\alpha}} + \frac{\left ( \sum_{i=1}^{2}l_{i}(t) \right )^{\frac{1}{\alpha} } }{\alpha-1}\right ) F_1^{1\gets }\left ( q \right),\\
&\MES^A(q,k,t)=\frac{\alpha }{\alpha-1} \frac{l_{k}(t)}{\left( \sum_{i=1}^{2}l_{i}(t) \right )^{1-\frac{1}{\alpha}}} F_1^{1\gets }\left ( q \right ), \, k=1,2,
\end{align*}
where $$l_k(t)=\left(\left(\frac{\gamma_k^1}{\gamma_1^1}\right)^{\alpha}\lambda_k^1+\left(\frac{\gamma_k^2}{\gamma_1^1}\right)^{\alpha}\lambda_k^2\right)\frac{e^{t\phi_R(\alpha)}-1}{\phi_R(\alpha)},\,\,k=1,2.$$

Then, we apply a Monte Carlo algorithm to simulate $\p(D_t>x)$, $\mathrm{SES}_{q,k}(D_{t})$ and $\mathrm{MES}_{q,k}(D_{t})$. 
For each group of parameters, after implementing a certain numerical algorithm for $N$ times, we get $N$ experimental results of the two business lines' discounted losses $\widetilde{Z_{k,t}^{i}},\,k=1,2$ and the total loss $\widetilde{D_{t}^{i}}$, where $i=1,2,\ldots,N$. Let $\widetilde{Z_{k,t}^{(1)}}\le\widetilde{Z_{k,t}^{(2)}}\le\cdots\le \widetilde{Z_{k,t}^{(N)}},\,k=1,2$ and $\widetilde{D_{t}^{(1)}}\le \widetilde{D_{t}^{(2)}}\le\cdots\le\widetilde{D_{t}^{(N)}}$ be the corresponding order statistics. Then the empirical estimates are: 
\begin{align*}
&p^E(x,t)=\frac{\sum_{i=1}^N\id_{\left\{\widetilde{D_t^i}>x\right\}}}{N},\\
&\SES^E(q,k,t)=\frac{\sum_{i=1}^{N}\left(\widetilde{Z_{k,t}^{i}}-\widetilde{Z_{k,t}^{(\lfloor Nq\rfloor)}}\right)^{+}\id_{\left\{\widetilde{D_{t}^{i}}>\widetilde{D_{t}^{(\lfloor Nq\rfloor)}}\right\}}}{\sum_{i=1}^{N}\id_{\left\{\widetilde{D_{t}^{i}}>\widetilde{D_{t}^{(\lfloor Nq\rfloor)}}\right\}}},\\
&\MES^E(q,k,t)=\frac{\sum_{i=1}^{N}\widetilde{Z_{k,t}^{i}}\id_{\left\{\widetilde{D_{t}^{i}}>\widetilde{D_{t}^{(\lfloor Nq\rfloor)}}\right\}}}{\sum_{i=1}^{N}\id_{\left\{\widetilde{D_{t}^{i}}>\widetilde{D_{t}^{(\lfloor Nq\rfloor)}}\right\}}},\,k=1,2,
\end{align*}
where for $a\in\R^{+}$, $\lfloor a\rfloor$ denotes the maximal integer that is smaller than $a$. Finally, we do a comparison between asymptotic approximations and empirical estimates to see the degree of consistency between them.
\begin{table}[htbp]
    \centering
    \begin{tabular}{ccccc}
    \toprule
     $x$ & Asymptotic approximation & Empirical estimate & $\frac{\text{Empirical}}{\text{Asymptotic}}$ \\
    \midrule
     100 & $1.3527\times 10^{-3}$ & $1.3605\times 10^{-3}$ & 1.0058 \\
     150 & $7.5376\times 10^{-4}$ & $7.5310\times 10^{-4}$ & 0.9991 \\
     200 & $4.9543\times 10^{-4}$ & $4.8320\times 10^{-4}$ & 0.9753 \\
     250 & $3.5705\times 10^{-4}$ & $3.5140\times 10^{-4}$ & 0.9842 \\
     300 & $2.7293\times 10^{-4}$ & $2.6890\times 10^{-4}$ & 0.9852 \\
    \bottomrule
    \end{tabular}
    \caption{Empirical and asymptotic values of $\p(D_{t}>x)$ w.r.t. $x$}
    \label{tab:my_label}
\end{table}
\begin{figure}[htbp] 
    \centering
    \begin{minipage}[t]{0.48\linewidth}
        \centering
        \includegraphics[width=7.9cm, height=7.2cm]{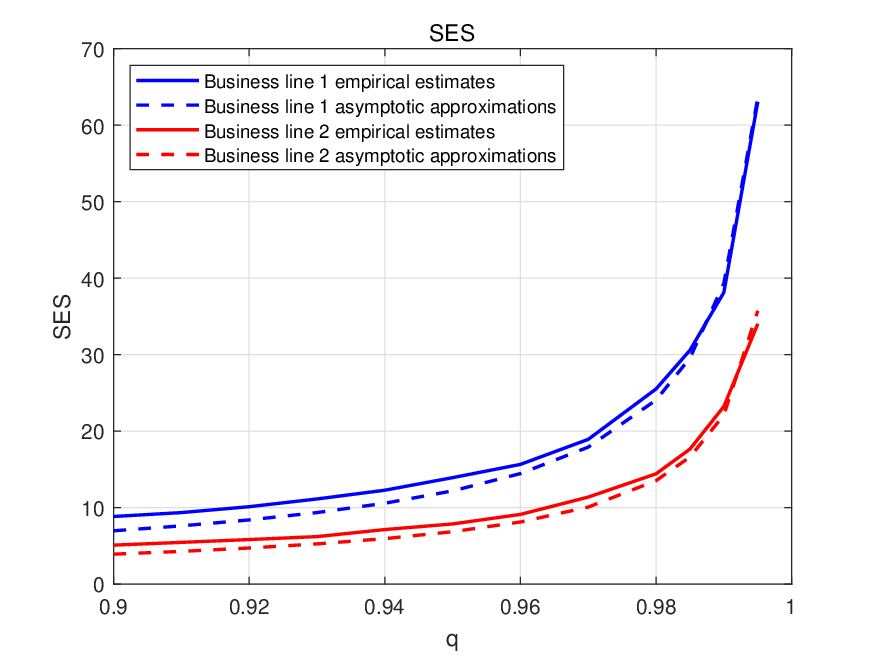}
    \end{minipage}
    \hfill
    \begin{minipage}[t]{0.48\linewidth}
        \centering
        \includegraphics[width=7.9cm, height=7.2cm]{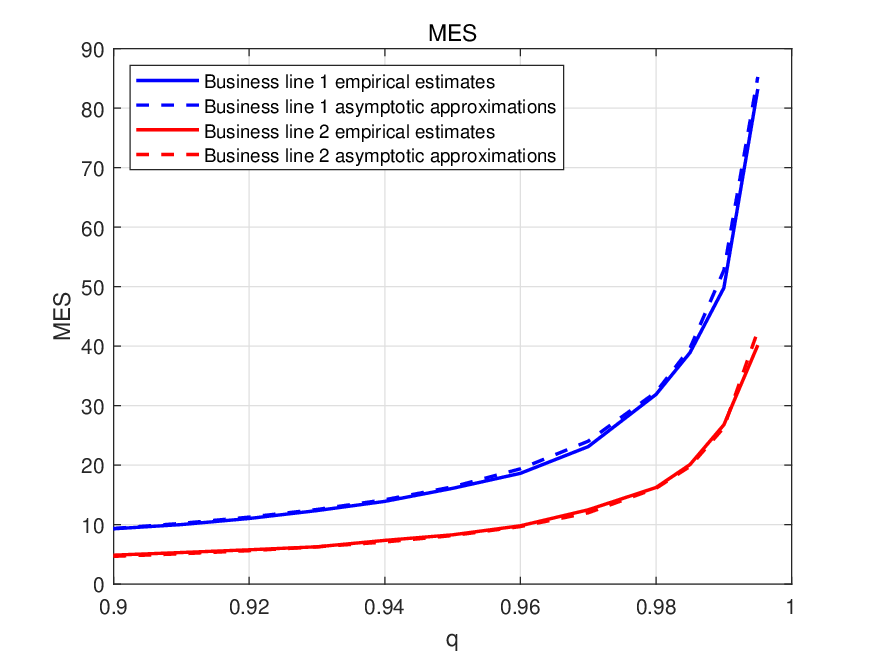}
    \end{minipage}
    \caption{Empirical and asymptotic values of $\text{SES}_{q,k}(D_{t})$ and $\text{MES}_{q,k}(D_{t})$ w.r.t. $q$}
    \label{fig:SES_MES}
\end{figure}


In this subsection, we set $N=10^6$, $\alpha=1.5$, $(\gamma_{1}^1,\gamma_{1}^2,\gamma_{2}^1,\gamma_{2}^2)=(4,6,2,4)$, $(\mu,\sigma,\lambda)=(10,0.2,0.5)$, $(\lambda_1^1,\lambda_1^2,\lambda_2^1,\lambda_2^2)=(0.4,0.7,0.4,0.7)$, $t=10$, $c_1=c_2=5$. For the simulation of Theorem \ref{thm1}, we take several $x$'s from 100 to 300, and the numerical results have been listed in Table \ref{tab:my_label}, where all the quotients of the empirical estimates and asymptotic approximations lie in the interval $[0.95,1.05]$. For the simulation of Theorem \ref{thm2}, we take several $q$'s equidistantly from 0.900 to 0.990 and an extra $q=0.995$ to illustrate the trend $q\uparrow 1$, and the results are plot in Figure \ref{fig:SES_MES}, where the curves of asymptotic approximations closely match those of empirical estimates. Therefore, the accuracy of our asymptotics is validated.

\subsection{Parameter impact and sensitivity analysis}
In this subsection, with our model set as in Subsection \ref{subsection4.1}, we analyze the impact of variation of parameters in the investment return process $\{R_t\}_{t\geq 0}$, the Poisson intensities $\lambda_k^j$, and the regularly varying index $\alpha$ on {$p^A(x,t)$, $\SES^A(q,k,t)$ and $\MES^A(q,k,t)$}. 

Suppose that the investment return process is given by \eqref{Rt}, where $R_t$ consists of a steadily increasing drift term $\mu t$, a diffusion term $\sigma W_t$ capturing market volatility, and a jump component $\sum_{i=1}^{N_t^{\prime}}\zeta_i$ reflecting sudden and significant changes. A direct analytical result demonstrates that the asymptotic approximations of $\p(D_t>x)$, $\SES_{q,k}(D_t)$, and $\MES_{q,k}(D_t)$ exhibit certain monotonic relationships with the drift, diffusion, and jump parameters of $R_t$, as stated in the following proposition.
\begin{proposition}\label{prop1}
Under the {setting} of Theorems \ref{thm1} and \ref{thm2}, suppose further the renewal processes are all homogeneous Poisson processes, and for $1\le k\le d,\,1\le j\le r$, $t\ge 0$, $\lambda_{k,t}^j=\lambda_k^jt$ where $\lambda_k^j$ are intensities of corresponding Poisson processes, and the investment return process $R_t$ is of the form (\ref{Rt}). Then, we have:\\
(i) the asymptotic approximations  $p^A(x,t)$, $\SES^A(q,k,t)$ and $\MES^A(q,k,t)$ are decreasing with $\mu$, given other parameters are fixed;\\
(ii) the aforementioned asymptotic approximations are increasing with  $\sigma$ or $\lambda$, given other parameters are fixed.
\end{proposition}

Indeed, Proposition \ref{prop1} aligns with the intuition: $\mu$ represents the return from stable investments, while $\sigma$ reflects market volatility and $\lambda$ reflects sudden change. A higher stable return will reduce the bankruptcy probability and expected capital shortfall of {a company}, and greater market fluctuations will lead to larger risk.

\begin{table}[htbp]
\centering
\small
\begin{threeparttable}
\begin{tabular}{lcccc}
\toprule
\multirow{2}{*}{} & \multirow{2}{*}{} & \multicolumn{2}{c}{\textbf{Percentage change w.r.t. the baseline}} \\
\cmidrule(lr){3-5}
& & $q=0.97$ & $q=0.98$ & $q=0.99$ \\
\midrule
\multirow{5}{*}{$\SES^A(q,1,t)$} 
    & +2\% & -14.03\% & -14.54\% & -15.42\% \\
    & +1\% & -7.45\% & -7.73\% & -8.21\% \\
    & $\alpha = 1.2$ & 0.00 & 0.00  & 0.00 \\
    & $-1\%$ & +8.49\% & +8.82\% & +9.41\% \\
    & $-2\%$ & +18.23\% & +18.97\% & +20.27\% \\
\midrule
\multirow{5}{*}{$\SES^A(q,2,t)$} 
    & +2\% & -14.44\% & -14.94\% & -15.82\% \\
    & +1\% & -7.67\% & -7.95\% & -8.43\% \\
    & $\alpha = 1.2$ & 0.00 & 0.00 & 0.00 \\
    & $-1\%$ & +8.75\% & +9.08\% & +9.67\% \\
    & $-2\%$ & +18.80\% & +19.54\% & +20.84\% \\
\midrule
\multirow{5}{*}{$|\SES^A(q,1,t)- \SES^A(q,2,t)|$} 
    & +2\% & -13.43\% & -13.94\% & -14.83\% \\
    & +1\% & -7.13\% & -7.40\% & -7.89\% \\
    & $\alpha = 1.2$ & 0.00 & 0.00 & 0.00 \\
    & $-1\%$ & +8.11\% & +8.44\% & +9.02\% \\
    & $-2\%$ & +17.40\% & +18.13\% & +19.42\% \\
\bottomrule
\end{tabular}
\caption{Asymptotic approximations for $\SES_{q,k}(D_t)$ at different $q$ values}
\label{table:ses_sens}
\begin{tablenotes}
\item 
\end{tablenotes}
\end{threeparttable}
\end{table}
\begin{table}[htbp]
\centering
\small
\begin{threeparttable}
\begin{tabular}{lcccc}
\toprule
\multirow{2}{*}{} & \multirow{2}{*}{} & \multicolumn{2}{c}{\textbf{Percentage change w.r.t. the baseline}} \\
\cmidrule(lr){3-5}
& & $q=0.97$ & $q=0.98$ & $q=0.99$ \\
\midrule
\multirow{5}{*}{$\MES^A(q,1,t)$} 
    & +2\% & -12.80\% & -13.31\% & -14.21\% \\
    & +1\% & -6.79\% & -7.06\% & -7.55\% \\
    & $\alpha = 1.2$ & 0.00 & 0.00  & 0.00 \\
    & $-1\%$ & +7.71\% & +8.04\% & +8.63\% \\
    & $-2\%$ & +16.55\% & +17.27\% & +18.55\% \\
\midrule
\multirow{5}{*}{$\MES^A(q,2,t)$} 
    & +2\% & -13.71\% & -14.21\% & -15.10\% \\
    & +1\% & -7.27\% & -7.55\% & -8.04\% \\
    & $\alpha = 1.2$ & 0.00 & 0.00 & 0.00 \\
    & $-1\%$ & +8.28\% & +8.61\% & +9.20\% \\
    & $-2\%$ & +17.78\% & +18.51\% & +19.81\% \\
\midrule
\multirow{5}{*}{$|\MES^A(q,1,t)- \MES^A(q,2,t)|$} 
    & +2\% & -11.60\% & -12.12\% & -13.03\% \\
    & +1\% & -6.14\% & -6.42\% & -6.91\% \\
    & $\alpha = 1.2$ & 0.00 & 0.00 & 0.00 \\
    & $-1\%$ & +6.97\% & +7.29\% & +7.87\% \\
    & $-2\%$ & +14.93\% & +15.64\% & +16.90\% \\
\bottomrule
\end{tabular}
\caption{Asymptotic approximations for $\MES_{q,k}(D_t)$ at different $q$ values}
\label{table:mes_sens}
\begin{tablenotes}
\item 
\end{tablenotes}
\end{threeparttable}
\end{table}
{Applying a similar analysis, we can see that given other parameters are fixed, the business lines with higher claim arrival rates (corresponding to Poisson intensities $\lambda_k^j,\,1\le k\le d,\,1\le j\le r$ in Proposition \ref{prop1}) will have higher {$p^A(x,t)$, $\SES^A(q,k,t)$ and $\MES^A(q,k,t)$} (reflecting the bankruptcy probability and expected capital shortfall). We now examine the impact of increased claim frequency on the absolute differences in the asymptotic approximations of SES (or MES) between different business lines. Following a similar proof as in Proposition \ref{prop1}, it can be shown that when all the intensities $\lambda_k^j,\,1\le k\le d,\,1\le j\le r$ increase proportionally (e.g., by $1\%$), the absolute differences---if nonzero initially---will increase. This implies that more frequent claims will exacerbate the disparity in capital shortfall among different business lines.}

As for the impact of changes in the regularly varying index $\alpha$, due to the complicated structure of the function between our asymptotic approximations and $\alpha$, it is quite hard to analytically examine the impact of variation of $\alpha$. Thus, we conduct {a} sensitivity analysis of $\SES^A(q,k,t)$, $\MES^A(q,k,t)$ and differences $|\SES^A(q,1,t)- \SES^A(q,2,t)|$ and $|\MES^A(q,1,t)- \MES^A(q,2,t)|$ with respect to the variation of $\alpha$. We set $\alpha=1.2$ as the benchmark, and apply small perturbations to $\alpha$ to examine how the aforementioned quantities change under different $q$. The results are presented in Tables \ref{table:ses_sens} and \ref{table:mes_sens}. As shown in these tables, for all values of $q$, those quantities decrease as $\alpha$ increases, and the rate of decrease is greater for a larger $q$, indicating a higher sensitivity to changes in $\alpha$.
}

\section{Proof of main results}\label{sec:5}
Before the proofs are stated, some remarks are worth noting. If asymptotic relations are considered only at each time point, which transfers the corresponding stochastic processes into certain random variables, the proof will be much easier. Actually, the one-dimensional case of (\ref{thm1_1}) at a fixed time point is just a simple corollary of Theorem 3.3 in \cite{chen2009sums}. The key difficulty is how to deal with uniform convergence. To address this issue, we split the complicated probabilities or expectations into simple components, and deal with each component individually. For each component, we isolate the ingredients that involve $t$ and demonstrate that the remaining parts of the variables converge independently of $t$. In the proof of Theorem \ref{thm2}, we employ a useful proposition (Proposition 0.5 of \cite{resnick2008extreme}), which demonstrates that, a regularly varying function can automatically hold the regularly varying property uniformly on some intervals. This property, together with the classic ``$\varepsilon$-$\delta$" definition of convergence which transfers asymptotic relations into inequalities, allows us to address uniform asymptotic relations.
\subsection{On Theorem 3.1}
In the proof of Theorem \ref{thm1}, we mainly refer to the methods of proving Theorem 3.1 in \cite{tang2010uniform} and \cite{li2012asymptotics}. Note that the setting of time-independence can be reduced from a time-dependent assumption in \cite{li2012asymptotics}. Hence, according to Lemma 3.9 of that paper, we have the following lemma. 

{
\begin{lemma}\label{lemma1}
Under the {setting} of Theorem \ref{thm1}, for all $i\geq 1$, all $1\le k\le d$, and all $1\le j\le r$, uniformly for all $t\in\Lambda$, we have
\begin{align*}
\p\left(X_{ki}^je^{-R_{\tau_{ki}^j}}\id_{\left \{ \tau_{ki}^j\le t \right \}}>x\right)\sim \overline{F_{k}^j}\left ( x \right )\int\limits_{0-}^{t}e^{s\phi \left ( \alpha \right)}\p\left(\tau_{ki}^j\in \mathrm{d}s\right).
\end{align*}    
\end{lemma}
}

Furthermore, we have the following result.

{
\begin{lemma}\label{lemmabu}
Under the {setting} of Theorem \ref{thm1}, for all $p,l\geq 1$, all $1\le k,h\le d$, and all $1\le a,b\le r$, uniformly for all $t\in\Lambda$, we have
\begin{align*}
&\quad\p\left( X_{kp}^ae^{-R_{\tau_{kp}^a}}\id_{\left \{ \tau_{kp}^a\le t \right \}}>x, X_{hl}^be^{-R_{\tau_{hl}^b}}\id_{\left \{ \tau_{hl}^b\le t \right \}}>x \right)\\
&= o(1)\left(\p\left( X_{kp}^ae^{-R_{\tau_{kp}^a}}\id_{\left \{ \tau_{kp}^a\le t \right \}}>x\right)+\p\left(X_{hl}^be^{-R_{\tau_{hl}^b}}\id_{\left \{ \tau_{hl}^b\le t \right \}}>x\right)\right).
\end{align*}
\end{lemma}
}
\begin{proof}
By basic set theory and {the} subadditivity of probability measures, for any positive function $a(x)$, we have 
\begin{align*}
&\quad\p\left( X_{kp}^ae^{-R_{\tau_{kp}^a}}\id_{\left \{ \tau_{kp}^a\le t \right \}}>x, X_{hl}^be^{-R_{\tau_{hl}^b}}\id_{\left \{ \tau_{hl}^b\le t \right \}}>x \right) \\
&\le \p\left( X_{kp}^ae^{-R_{\tau_{kp}^a}}\id_{\left \{ \tau_{kp}^a\le t \right \}}>x, X_{hl}^be^{-R_{\tau_{hl}^b}}\id_{\left \{ \tau_{hl}^b\le t \right \}}>x, e^{-R_{\tau_{kp}^a}}\id_{\left \{ \tau_{kp}^a\le t \right \}}\le {a(x)}, e^{-R_{\tau_{hl}^b}}\id_{\left \{ \tau_{hl}^b\le t \right \}}\le{a(x)} \right) \\
&\quad+ \p\left(e^{-R_{\tau_{kp}^a}}\id_{\left \{ \tau_{kp}^a\le t \right \}}>{a(x)}\right) + \p\left(e^{-R_{\tau_{hl}^b}}\id_{\left \{ \tau_{hl}^b\le t \right \}}>{a(x)}\right)\\
&=: I_{1}(x,t) + I_{2}(x,t) + I_{3}(x,t).
\end{align*}
And according to Lemma 3.3 in \cite{li2012asymptotics}, there exists a positive function $a(x)$ satisfying
$a(x)\rightarrow\infty$ and $a(x) = o(x)$, such that it holds uniformly for all $t\in\Lambda$ that
\begin{align}
I_{2}(x,t)=o(1)\p\left(X_{kp}^ae^{-R_{\tau_{kp}^a}}\id_{\left \{ \tau_{kp}^a\le t \right \}}>x\right),\:I_{3}(x,t)=o(1)\p\left(X_{hl}^be^{-R_{\tau_{hl}^b}}\id_{\left \{ \tau_{hl}^b\le t \right \}}>x\right).\label{I1,I2}
\end{align}
For $I_{1}(x,t)$, since $X_{kp}$ and $X_{hl}$ are AI, for all $\varepsilon>0$, there is an $M>0$ such that for all $x^{\ast}>M$, $$\p\left(X_{kp}^a>x^{\ast},X_{hl}^b>x^{\ast}\right)\le \varepsilon \p\left(X_{kp}^a>x^{\ast}\right).$$ 
Take large enough $x$ such that $x/a(x)>M$. For all $t\in\Lambda$, suppose $G_{t}(\cdot,\cdot)$ is the joint distribution of $e^{-R_{\tau_{kp}^a}}\id_{\left \{ \tau_{kp}^a\le t \right \}}$ and $e^{-R_{\tau_{hl}^b}}\id_{\left \{ \tau_{hl}^b\le t \right \}}$. Then, we have
\begin{align}
I_{1}(x,t)&=\iint\limits_{0<w_{1},w_{2}\le a(x)}\p\left(w_{1}X_{kp}^a>x,w_{2}X_{hl}^b>x\right)G_{t}(\d w_{1},\d w_{2}) \nonumber\\
&\le\varepsilon \iint\limits_{0<w_{1},w_{2}\le a(x)}\p\left(w_{1}X_{kp}^a>x\right)G_{t}(\mathrm{d}w_{1},\mathrm{d}w_{2}) \nonumber \\
&\le \varepsilon \p\left(X_{kp}^ae^{-R_{\tau_{kp}^a}}\id_{\left \{ \tau_{kp}^a\le t \right \}}>x\right).\label{lemma5.4.2}
\end{align}
Therefore, by (\ref{I1,I2}) and (\ref{lemma5.4.2}), we can get the result.
\end{proof}

Further, using Lemmas \ref{lemma1} and \ref{lemmabu}, we can derive the following two lemmas.
{
\begin{lemma}\label{lemma3}
Under the {setting} of Theorem \ref{thm1}, for all $m=1,2,{\ldots}$, it holds uniformly for all $t \in \Lambda$ that
\begin{align*}
\quad\p\left ( \sum_{k=1}^{d}\sum_{j=1}^{r}\sum_{i=1}^{m}X_{ki}^je^{-R_{\tau_{ki}^j}}\id_{\left \{ \tau_{ki}^j\le t \right \}}> x \right )\sim \sum_{k=1}^{d}\sum_{j=1}^{r} \overline{F_{k}^j}\left ( x \right )\int\limits_{0-}^{t}e^{s\phi \left ( \alpha \right)}\sum_{i=1}^{m}\p\left(\tau_{ki}^j\in \mathrm{d}s\right).  \end{align*}    
\end{lemma}
}
\begin{proof}
In the following proof, we only consider the case when $d=1,\,r=2$ for simplicity, and more general {cases} can be derived similarly. We use the notation:$$H_{t}:=\sum_{j=1}^2\sum_{i=1}^{m}X_{1i}^je^{-R_{\tau_{1i}^j}}\id_{\left \{ \tau_{1i}^j\le t \right \}}.$$ That is, we need to prove that it uniformly holds for all $t\in \Lambda$ that
\begin{align}
\p\left(H_{t}>x\right) \sim \sum_{j=1}^2\sum_{i=1}^{m}\p\left(X_{1i}^je^{-R_{\tau_{1i}^j}}\id_{\left \{ \tau_{1i}^j\le t \right \}}>x\right),\label{lemma5.4.1}
\end{align}
and then the conclusion comes naturally from Lemma \ref{lemma1}.
First, notice that
\begin{align*}
\p\left(H_{t}>x\right) &\ge \p\left(\bigcup_{j=1}^2\bigcup_{i=1}^{m}\left\{X_{1i}^je^{-R_{\tau_{1i}^j}}\id_{\left \{ \tau_{1i}^j\le t \right \}}>x\right\}\right) \\
&\ge \sum_{j=1}^2\sum_{i=1}^{m}\p\left(X_{1i}^je^{-R_{\tau_{1i}^j}}\id_{\left \{ \tau_{1i}^j\le t \right \}}>x\right) \\
&\quad- \sum_{j=1}^2\sum_{1\le p <{l}\le m}\p\left( X_{1p}^je^{-R_{\tau_{1p}^j}}\id_{\left \{ \tau_{1p}^j\le t \right \}}>x, X_{1l}^je^{-R_{\tau_{1l}^j}}\id_{\left \{ \tau_{1l}^j\le t \right \}}>x \right)\\
&\quad-\sum_{1\le p,{l}\le m}\p\left( X_{1p}^1e^{-R_{\tau_{1p}^1}}\id_{\left \{ \tau_{1p}^1\le t \right \}}>x, X_{1l}^2e^{-R_{\tau_{1l}^2}}\id_{\left \{ \tau_{1l}^2\le t \right \}}>x \right).
\end{align*}
Therefore, the lower-bound direction of (\ref{lemma5.4.1}) is achieved through Lemma \ref{lemmabu}. Next, consider the upper-bound direction of (\ref{lemma5.4.1}). For any fixed $0<\delta<1$, we have
\begin{align*}
&\quad \p\left(H_{t}>x\right)\\
&\le \p\left(\bigcup_{j=1}^2\bigcup_{i=1}^{m}\left\{X_{1i}^je^{-R_{\tau_{1i}^j}}\id_{\left \{ \tau_{1i}^j\le t \right \}}>(1-\delta)x\right\}\right) \\
&\quad+ \p\left(H_{t}>x,\bigcup_{j=1}^2\bigcup_{i=1}^{m}\left\{X_{1i}^je^{-R_{\tau_{1i}^j}}\id_{\left \{ \tau_{1i}^j\le t \right \}}>\frac{x}{2m}\right\},\bigcap_{j=1}^2\bigcap_{i=1}^{m}\left\{X_{1i}^je^{-R_{\tau_{1i}^j}}\id_{\left \{ \tau_{1i}^j\le t \right \}}\le(1-\delta)x\right\}\right)\\
&=:J_{1}(x,t)+J_{2}(x,t).
\end{align*}
For $J_{1}(x,t)$, by Lemma \ref{lemma1}, uniformly for all $t\in\Lambda$, we have
\begin{align*}
J_{1}(x,t)&\le\sum_{j=1}^2\sum_{i=1}^{m}\p\left(X_{1i}^je^{-R_{\tau_{1i}^j}}\id_{\left \{ \tau_{1i}^j\le t \right \}}>(1-\delta)x\right)\\
&\sim(1-\delta)^{-\alpha}\sum_{j=1}^2\overline{F_{1}^j}\left ( x \right )\int\limits_{0-}^{t}e^{s\phi \left ( \alpha \right)}\sum_{i=1}^{m}\p\left(\tau_{1i}^j\in \mathrm{d}s\right).
\end{align*}
For $J_{2}(x,t)$, by Lemma \ref{lemmabu}, uniformly for all $t\in\Lambda$, we have
\begin{align*}
J_{2}(x,t)&\le \sum_{j=1}^2\sum_{i=1}^{m}\p\left(X_{1i}^je^{-R_{\tau_{1i}^j}}\id_{\left \{ \tau_{1i}^j\le t \right \}}>\frac{x}{2m},H_{t}-X_{1i}^je^{-R_{\tau_{1i}^j}}\id_{\left \{ \tau_{1i}^j\le t \right \}}> \delta x\right)\\
&\le \sum_{j=1}^2\sum_{1\le p\ne l\le m}\p\left(X_{1p}^je^{-R_{\tau_{1p}^j}}\id_{\left \{ \tau_{1p}^j\le t \right \}}>\frac{\delta x}{2m},X_{1l}^je^{-R_{\tau_{1l}^j}}\id_{\left \{ \tau_{1l}^j\le t \right \}}>\frac{\delta x}{2m}\right) \\
&\quad+2\sum_{p=1}^{m}\sum_{l=1}^{m}\p\left( X_{1p}^1e^{-R_{\tau_{1p}^1}}\id_{\left \{ \tau_{1p}^1\le t \right \}}>\frac{\delta x}{2m},X_{1l}^2e^{-R_{\tau_{1l}^2}}\id_{\left \{ \tau_{1l}^2\le t \right \}}>\frac{\delta x}{2m}\right)\\
&=o(1)\sum_{j=1}^2\sum_{i=1}^{m}\p\left(X_{1i}^je^{-R_{\tau_{1i}^j}}\id_{\left \{ \tau_{1i}^j\le t \right \}}>x\right).
\end{align*}
Therefore, combined with Lemma \ref{lemma1}, we get (\ref{lemma5.4.1}), and then this lemma is proven.    
\end{proof}

{
\begin{lemma}\label{lemma4}
Under the 
{setting} of Theorem \ref{thm1}, for all $0<\varepsilon<1$, there is some large integer $m^{\prime}$ such that for all $m\ge m^{\prime}$, it holds uniformly for $t\in\Lambda$ that
\begin{align*}
\p\left ( \sum_{k=1}^{d}\sum_{j=1}^{r}\sum_{i=m+1}^{\infty}X_{ki}^je^{-R_{\tau_{ki}^j}}\id_{\left \{ \tau_{ki}^j\le t \right \}}> x \right ) \lesssim \varepsilon\sum_{k=1}^{d}\sum_{j=1}^{r}\overline{F_{k}^j}\left ( x \right )\int\limits_{0-}^{t}e^{s\phi \left ( \alpha \right ) }\mathrm{d}\lambda_{k,s}^{j}  .
\end{align*}    
\end{lemma}
}
\begin{proof}
According to the subadditivity of the probability measure, we can get
\begin{align*} 
\p\left ( \sum_{k=1}^{d}\sum_{j=1}^{r}\sum_{i=m+1}^{\infty}X_{ki}^je^{-R_{\tau_{ki}^j}}\id_{\left \{ \tau_{ki}^j\le t \right \}}> x \right ) \nonumber&\le \p\left ( \bigcup_{k=1}^{d}\bigcup_{j=1}^{r}\left \{ \sum_{i=m+1}^{\infty}X_{ki}^je^{-R_{\tau_{ki}^j}}\id_{\left \{ \tau_{ki}^j\le t \right \}} > \frac{x}{dr} \right \} \right )  \nonumber\\
&\le \sum_{k=1}^{d}\sum_{j=1}^{r}\p\left ( \sum_{i=m+1}^{\infty}X_{ki}^je^{-R_{\tau_{ki}^j}}\id_{\left \{ \tau_{ki}^j\le t \right \}} > \frac{x}{dr} \right).
\end{align*}
 Now imitating the proof of Lemma 3.8 in \cite{li2012asymptotics}, we can get that for all $1 \le k \le d$ and all $1\le j\le r$, there is some large integer $m^{\prime}$ such that for all $m\ge m^{\prime}$, we have that uniformly for all $t\in\Lambda$,
\begin{align*}
\p\left( \sum_{i=m+1}^{\infty}X_{ki}^je^{-R_{\tau_{ki}^j}}\id_{\left \{ \tau_{ki}^j\le t \right \}} > \frac{x}{dr} \right ) \lesssim \frac{\varepsilon}{dr} \overline{F_{k}^j}\left ( x \right )\int\limits_{0-}^{t}e^{s\phi \left ( \alpha \right ) }\mathrm{d}\lambda _{k,s}^{j}. 
\end{align*}
Thus, the required conclusion is proven.
\end{proof}

Now, we can state our proof of Theorem \ref{thm1}.\\
\textbf{Proof of Theorem \ref{thm1}.} ({\romannumeral1}) Notice that for all $m\in \mathbb{N}$ and all fixed $0<\beta<1$,
\begin{align*}
\p\left(S_{t} > x \right) &\le \p\left ( \sum_{k=1}^{d}\sum_{j=1}^{r}\sum_{i=1}^{m}X_{ki}^je^{-R_{\tau_{ki}^j}}\id_{\left \{ \tau_{ki}^j\le t \right \}}> (1-\beta) x \right )\\
&+ \p\left ( \sum_{k=1}^{d}\sum_{j=1}^{r}\sum_{i=m+1}^{\infty}X_{ki}^je^{-R_{\tau_{ki}^j}}\id_{\left \{ \tau_{ki}^j\le t \right \}}> \beta x \right )\\&=: K_{1}(x,t)+K_{2}(x,t).
\end{align*}
Now by Lemmas \ref{lemma3} and \ref{lemma4}, for all $\varepsilon>0$, all 
 large $m$, it holds uniformly for $t\in\Lambda$ that
\begin{align*}
K_{1}(x,t)\sim (1-\beta)^{-\alpha}\sum_{k=1}^{d}\sum_{j=1}^{r} \overline{F_{k}^j}\left ( x \right )\int\limits_{0-}^{t}e^{s\phi \left ( \alpha \right)}\sum_{i=1}^{m}\p\left(\tau_{ki}^j\in \mathrm{d}s\right)\le (1-\beta)^{-\alpha}\sum_{k=1}^{d}\sum_{j=1}^{r} \overline{F_{k}^j}\left ( x \right )\int\limits_{0-}^{t}e^{s\phi \left ( \alpha \right)}\mathrm{d}\lambda_{k,s}^j
\end{align*} 
and
\begin{align*}
K_{2}(x,t)\lesssim \beta^{-\alpha}\varepsilon \sum_{k=1}^{d}\sum_{j=1}^{r} \overline{F_{k}^j}\left ( x \right )\int\limits_{0-}^{t}e^{s\phi \left ( \alpha \right)}\mathrm{d}\lambda_{k,s}^j .
\end{align*}
Then it holds uniformly for $t\in\Lambda$ that $$K_{1}(x,t)+K_{2}(x,t)\lesssim \left((1-\beta)^{-\alpha}+\beta^{-\alpha}\varepsilon\right) \sum_{k=1}^{d}\sum_{j=1}^{r} \overline{F_{k}^j}\left ( x \right )\int\limits_{0-}^{t}e^{s\phi \left ( \alpha \right)}\mathrm{d}\lambda_{k,s}^j.$$ By the  arbitrariness of $\beta$ and $\varepsilon$, we get the upper-bound direction result of (\ref{thm1_1}). For the lower-bound direction, first notice that for all $1\le k\le d$, all $1\le j\le r$, all $m\in \mathbb{N}$, all $t\in\Lambda$ and $t<\infty$, we have
\begin{align}
\overline{F_{k}^j}\left( x \right)\int\limits_{0-}^{t}e^{s\phi \left ( \alpha \right)}\sum_{i=1}^{m}\p\left(\tau_{ki}^j\in \mathrm{d}s\right) &= \overline{F_{k}^j}(x)\int\limits_{0-}^{t}e^{s\phi \left ( \alpha \right)}\left(\mathrm{d}\lambda_{k,s}^{j}-\sum_{i=m+1}^{\infty}\p\left(\tau_{ki}^j\in \mathrm{d}s\right)\right) \nonumber\\
&\ge \overline{F_{k}^j}(x)\left(\int\limits_{0-}^{t}e^{s\phi \left ( \alpha \right)}\mathrm{d}\lambda_{k,s}^{j}-\E\left[N_{k,t}^{j}\id_{\left\{N_{k,t}^{j}\ge m+1\right\}}\right]\right).\label{3.5}
\end{align}
Since $m$ is arbitrary, according to Lemma 5.3 of \cite{tang2004uniform}, for an arbitrary $T\in\Lambda$, $\E\left[N_{k,t}^{j}\id_{\left\{N_{k,t}^{j}\ge m+1\right\}}\right]$ tends to 0 uniformly for $t\in\Lambda_{T}$ as $m\to\infty$. By Lemma \ref{lemma3}, the following relation
\begin{align*}
\p\left(S_{t}>x\right)&\ge \p\left ( \sum_{k=1}^{d}\sum_{j=1}^{r}\sum_{i=1}^{m}X_{ki}^je^{-R_{\tau_{ki}^j}}\id_{\left \{ \tau_{ki}^j\le t \right \}}> x \right )\sim \sum_{k=1}^{d}\sum_{j=1}^{r} \overline{F_{k}^j}\left ( x \right )\int\limits_{0-}^{t}e^{s\phi \left ( \alpha \right)}\sum_{i=1}^{m}\p\left(\tau_{ki}^j\in \mathrm{d}s\right)
\end{align*}
holds uniformly for all $t\in\Lambda$. Then by (\ref{3.5}), we get that uniformly for all $t\in\Lambda_{T}$, 
\begin{align}
\p\left(S_{t}>x\right) \gtrsim \sum_{k=1}^{d}\sum_{j=1}^{r} \overline{F_{k}^j}\left ( x \right )\int\limits_{0-}^{t}e^{s\phi \left ( \alpha \right)}\mathrm{d}\lambda_{k,s}^j.\label{p(St>x)>}    
\end{align} 
Next, by \eqref{integral<inf}, for all $1\le k\le d$ and $1\le j\le r$, we have $0<\int\limits_{0-}^{\infty}e^{s\phi \left ( \alpha \right)}\mathrm{d}\lambda_{k,s}^{j}<\infty$. Hence, for all $\varepsilon>0$, we can find some $T\in\Lambda$ such that $\int\limits_{T}^{\infty}e^{s\phi \left ( \alpha \right)}\mathrm{d}\lambda_{k,s}^{j}\le\varepsilon\int\limits_{0-}^{T}e^{s\phi \left ( \alpha \right)}\mathrm{d}\lambda_{k,s}^{j}$ for all $k,j$. Therefore, by \eqref{p(St>x)>}, we get that for any $t\in\Lambda^{T}$,
\begin{align*}
\p\left(S_{t}>x\right)\ge\p\left(S_{T}>x\right)&\gtrsim\sum_{k=1}^{d}\sum_{j=1}^{r} \overline{F_{k}^j}\left ( x \right )\int\limits_{0-}^{T}e^{s\phi \left ( \alpha \right)}\mathrm{d}\lambda_{k,s}^j\\
&\ge\frac{1}{1+\varepsilon}\sum_{k=1}^{d}\sum_{j=1}^{r} \overline{F_{k}^j}\left ( x \right )\int\limits_{0-}^{\infty}e^{s\phi \left ( \alpha \right)}\mathrm{d}\lambda_{k,s}^j\\
&\ge \frac{1}{1+\varepsilon}\sum_{k=1}^{d}\sum_{j=1}^{r} \overline{F_{k}^j}\left ( x \right )\int\limits_{0-}^{t}e^{s\phi \left ( \alpha \right)}\mathrm{d}\lambda_{k,s}^j.
\end{align*}
Letting $\varepsilon\downarrow 0$, we can get that (\ref{p(St>x)>}) holds uniformly for all $t\in\Lambda^{T}$, and thus uniformly for all $t\in\Lambda$.

({\romannumeral2}) Let $c:=\sum_{k=1}^{d}c_{k}$. Firstly, we consider the upper-bound direction of (\ref{thm1_2}). By ({\romannumeral1}), it holds uniformly for all $t\in\Lambda$ that 
\begin{align*}
\p(D_{t}>x)=\p\left(S_{t}-c\int\limits_{0}^{t}e^{-R_{u}}\mathrm{d}u > x\right)\le \p\left(S_{t}>x\right)\sim \sum_{k=1}^{d}\sum_{j=1}^{r} \overline{F_{k}^j}\left ( x \right )\int\limits_{0-}^{t}e^{s\phi \left ( \alpha \right)}\mathrm{d}\lambda_{k,s}^j.   
\end{align*}
Then consider the lower-bound direction of (\ref{thm1_2}). For any fixed $\delta>0$ and all $t\in\Lambda$, we have
\begin{align*}
    \p\left(S_{t}-c\int\limits_{0}^{t}e^{-R_{u}}\mathrm{d}u > x\right)&\ge \p\left(S_{t}-c\int\limits_{0}^{\infty}e^{-R_{u}}\mathrm{d}u > x\right)\nn\\
    &\ge \p\left(S_{t}>(1+\delta)x\right)-\p\left(c\int\limits_{0}^{\infty}e^{-R_{u}}\mathrm{d}u > \delta x\right).
\end{align*}
Due to ({\romannumeral1}), uniformly for all $t\in\Lambda$, we have $$\p\left(S_{t}>(1+\delta)x\right)\sim (1+\delta)^{-\alpha}\sum_{k=1}^{d}\sum_{j=1}^{r} \overline{F_{k}^j}\left ( x \right )\int\limits_{0-}^{t}e^{s\phi \left ( \alpha \right)}\mathrm{d}\lambda_{k,s}^j.$$ Now according to Lemma 4.6 of \cite{tang2010uniform}, we have $\E\left[\left(\int\limits_{0}^{\infty}e^{-R_{u}}\mathrm{d}u\right)^{\alpha^{\ast}}\right]<\infty$. So by Markov's inequality and (\ref{pre1}), we get that for all $T\in\Lambda$, it holds uniformly for all $t\in\Lambda^{T}$ that 
\begin{align*}
\p\left(c\int\limits_{0}^{\infty}e^{-R_{u}} \mathrm{d}u > \delta x\right)&\le \left(\frac{\delta x}{c}\right)^{-\alpha^{\ast}}\E\left[\left(\int\limits_{0}^{\infty}e^{-R_{u}} \mathrm{d}u\right)^{\alpha^{\ast}}\right]\\
&=o(1)\overline{F_{1}^1}\left ( x \right )\int\limits_{0-}^{t}e^{s\phi \left ( \alpha \right ) } \d\lambda _{1,s}^{1}\\
&=o(1)\sum_{k=1}^{d}\sum_{j=1}^{r} \overline{F_{k}^j}\left ( x \right )\int\limits_{0-}^{t}e^{s\phi \left ( \alpha \right)}\mathrm{d}\lambda_{k,s}^j,  
\end{align*}
where the second step is due to the fact that $$0<\int\limits_{0-}^{T}e^{s\phi \left ( \alpha \right ) } \d\lambda _{1,s}^{1}\le\int\limits_{0-}^{t}e^{s\phi \left ( \alpha \right ) } \d\lambda _{1,s}^{1}\le\int\limits_{0-}^{\infty}e^{s\phi \left ( \alpha \right ) } \d\lambda _{1,s}^{1}<\infty.$$ Then by the arbitrariness of $\delta$, we get that the lower-bound result holds uniformly for all $t\in\Lambda^{T}$. Therefore, Theorem \ref{thm1} has been proven.
 $\blacksquare$

\subsection{On Theorem 3.2}

According to Theorem \ref{thm1}(\romannumeral2), for all $T\in\Lambda$, it holds uniformly for all $t\in\Lambda^{T}$ that $\p(D_{t}>x)\sim\sum_{i=1}^{d}l_{i}(t)\overline{F}(x)$ as $x\rightarrow\infty$. Also notice that $\sum_{i=1}^{d}l_{i}(t)$ is increasing and continuous for $t\in [T,\infty)$, and $$\inf_{t\in [T,\infty)}\sum_{i=1}^{d}l_{i}(t)=\sum_{i=1}^{d}l_{i}(T)=:B>0,\, \sup_{t\in [T,\infty)}\sum_{i=1}^{d}l_{i}(t)=\sum_{i=1}^{d}l_{i}(\infty)<\infty,$$ so the range of $\sum_{i=1}^{d}l_{i}(t)$ on $[T,\infty]$ is actually a bounded closed interval of $\mathbb{R^{+}}$.

\begin{lemma}\label{lemma5}
Under the {setting} of Theorem \ref{thm2}, for any fixed $T\in\Lambda$, it holds uniformly for all $t\in\Lambda^{T}$ that
\begin{align}
F_{D_{t}}^{\gets }\left ( q \right ) \sim\left(\sum_{i=1}^{d}l_{i}(t)\right)^{\frac{1}{\alpha}} F^{\gets} \left ( q \right ).\label{uniform1}
\end{align}   
\end{lemma}
\begin{proof}
For a fixed $0<\varepsilon<1$, by Theorem \ref{thm1}, there is some $x^{\prime}>0$ such that for all $t\in\Lambda^{T}$ and all $x\ge x^{\prime}$, $$(1-\varepsilon)\sum_{i=1}^{d}l_{i}(t)\overline{F}(x)\le \p(D_{t}>x)\le (1+\varepsilon)\sum_{i=1}^{d}l_{i}(t)\overline{F}(x).$$ Take $q>1-(1-\varepsilon)B\overline{F}(x^{\prime})$. Then, $$1-q<(1-\varepsilon)B\overline{F}(x^{\prime})\le (1-\varepsilon)\sum_{i=1}^{d}l_{i}(t)\overline{F}(x^{\prime})\le \p(D_{t}>x^{\prime})\le (1+\varepsilon)\sum_{i=1}^{d}l_{i}(t)\overline{F}(x^{\prime}).$$ Hence, we have 
\begin{align*}
  x^{\prime}\le \inf\left\{y:(1-\varepsilon)\sum_{i=1}^{d}l_{i}(t)\overline{F}(y)\le 1-q\right\}&\le\inf\left\{y:\p(D_{t}>y)\le 1-q\right\}\\
  &\le \inf\left\{y:(1+\varepsilon)\sum_{i=1}^{d}l_{i}(t)\overline{F}(y)\le 1-q\right\}.  
\end{align*}
By Proposition 0.8(v) of \cite{resnick2008extreme}, $\left(\frac{1}{\overline{F}}\right)^{\gets}\in\RV_{\frac{1}{\alpha}}$, and by Proposition 0.5 of \cite{resnick2008extreme}, it naturally holds uniformly for all $t\in\Lambda^{T}$ that $$\inf\left\{y:(1+\varepsilon)\sum_{i=1}^{d}l_{i}(t)\overline{F}(y)\le 1-q\right\}\sim\left((1+\varepsilon)\sum_{i=1}^{d}l_{i}(t)\right)^{\frac{1}{\alpha}}F^{\gets} \left ( q \right ),$$ and the same holds for the side with $1-\varepsilon$. Letting $\varepsilon\downarrow 0$, we get the result.    
\end{proof}

\begin{lemma}\label{lemma6}
Under the {setting} of Theorem \ref{thm2}, for any fixed $T\in\Lambda$, it holds uniformly for all $t\in\Lambda^{T}$ that
\begin{align}
&\p\left ( Z_{k,t}> F_{D_{t}}^{\gets }\left ( q \right ) \right )\sim \frac{l_{k}(t)}{\sum_{i=1}^{d}l_{i}(t)}(1-q);\label{lemma5.7.1}\\
&\E\left [ \left ( Z_{k,t}-F_{D_{t}}^{\gets }\left ( q \right ) \right )^{+} \right ]\sim\frac{1-q}{\alpha-1} \frac{l_{k}(t)}{\left( \sum_{i=1}^{d}l_{i}(t) \right)^{1-\frac{1}{\alpha}}} F^{\gets }\left ( q \right );\label{lemma5.7.2}\\
&\E\left [ Z_{k,t}\id_{\left \{Z_{k,t}>F_{D_{t}}^{\gets }\left ( q \right ) \right \}} \right ]\sim\frac{\alpha \left(1-q \right)}{\alpha-1} \frac{l_{k}(t)}{\left( \sum_{i=1}^{d}l_{i}(t) \right)^{1-\frac{1}{\alpha}}} F^{\gets }\left ( q \right ).\label{lemma5.7.3}
\end{align}    
\end{lemma}
\begin{proof}
For \eqref{lemma5.7.1}, by Lemma \ref{lemma5}, fixing some $T\in\Lambda$, we have that for all $t\in\Lambda^{T}$, for any $0<\varepsilon<1$, for large enough $q$ independent of $t$, $$F_{D_{t}}^{\gets }\left ( q \right )\ge (1-\varepsilon)\left(\sum_{i=1}^{d}l_{i}(t)\right)^{\frac{1}{\alpha}} F^{\gets} \left ( q \right ),$$ which implies that for all $t\in\Lambda^{T}$,$$\p\left ( Z_{k,t}> F_{D_{t}}^{\gets }\left ( q \right ) \right )\le \p\left(Z_{k,t}>(1-\varepsilon)\left(\sum_{i=1}^{d}l_{i}(t)\right)^{\frac{1}{\alpha}} F^{\gets} \left ( q \right )\right).$$ Now following a similar way of proving Theorem \ref{thm1}, we can get that it holds uniformly for 
 all $t\in\Lambda^{T}$ that
\begin{align}
\p(Z_{k,t}>x)\sim l_{k}(t)\overline{F}(x). \label{uniform2}
\end{align}
That is, for all $\varepsilon^{\prime}>0$, there is some $x^{\prime\prime}>0$ such that for all $t\in\Lambda^{T}$ {and} all $x\ge x^{\prime\prime}$, $$\p(Z_{k,t}>x)\le (1+\varepsilon^{\prime})l_{k}(t)\overline{F}(x).$$ Take large enough $q$ such that $$x^{\prime\prime}\le (1-\varepsilon)B^{\frac{1}{\alpha}}F^{\gets} \left ( q \right )\le (1-\varepsilon)\left(\sum_{i=1}^{d}l_{i}(t)\right)^{\frac{1}{\alpha}}F^{\gets} \left ( q \right ).$$ Then, we have
\begin{align*}
\p\left(Z_{k,t}>(1-\varepsilon)\left(\sum_{i=1}^{d}l_{i}(t)\right)^{\frac{1}{\alpha}} F^{\gets} \left ( q \right )\right)&\le (1+\varepsilon^{\prime})l_{k}(t)\overline{F}\left((1-\varepsilon)\left(\sum_{i=1}^{d}l_{i}(t)\right)^{\frac{1}{\alpha}} F^{\gets} \left ( q \right )\right)\\
&\sim (1+\varepsilon^{\prime})(1-\varepsilon)^{-\alpha}\frac{l_{k}(t)}{\sum_{i=1}^{d}l_{i}(t)}\overline{F}(F^{\gets} \left ( q \right ))\\
&\sim (1+\varepsilon^{\prime})(1-\varepsilon)^{-\alpha}\frac{l_{k}(t)}{\sum_{i=1}^{d}l_{i}(t)}(1-q)
\end{align*}
 holds uniformly for all $t\in\Lambda^{T}$, where the first asymptotic relation comes from the regularly varying property of $F$ and Proposition 0.5 of \cite{resnick2008extreme}, and the last one comes from the fact $\overline{F}(F^{\gets} \left ( q \right ))\sim 1-q$. So letting $\varepsilon,\varepsilon^{\prime}\downarrow 0$, we get the upper-bound direction version of (\ref{lemma5.7.1}). The lower-bound direction is derived through a similar way.

To get (\ref{lemma5.7.2}), by uniform asymptotic relations (\ref{uniform1}) and (\ref{uniform2}), for all $0<\varepsilon,\varepsilon^{\prime}<1$ {and} all $t\in\Lambda^{T}$, there exists an appropriate $q^{\prime}$ such that for all $q>q^{\prime}$, we have
\begin{align}
\E\left [ \left ( Z_{k,t}-F_{D_{t}}^{\gets }\left ( q \right ) \right )^{+} \right ]&=\int\limits_{F_{D_{t}}^{\gets }\left ( q \right )}^{\infty}\p\left (Z_{k,t}>z \right )\mathrm{d}z\nonumber\\
&\le\int\limits_{(1-\varepsilon)\left(\sum_{i=1}^{d}l_{i}(t)\right)^{\frac{1}{\alpha}} F^{\gets} \left ( q \right )}^{\infty}\p\left (Z_{k,t}>z \right )\mathrm{d}z\nonumber\\
&\le (1+\varepsilon^{\prime})l_{k}(t)\int\limits_{(1-\varepsilon)\left(\sum_{i=1}^{d}l_{i}(t)\right)^{\frac{1}{\alpha}} F^{\gets} \left ( q \right )}^{\infty}\overline{F}(z)\mathrm{d}z.\label{lemma5.7pf1}
\end{align}
Moreover, Karamata's Theorem (see Theorem 0.6 in \cite{resnick2008extreme}) indicates that for all $\varepsilon^{\prime\prime}>0$, there is an $x^{\prime\prime\prime}>0$ such that for all $x\ge x^{\prime\prime\prime}$, $$\int\limits_{x}^{\infty}\overline{F}(z)\mathrm{d}z\le (1+\varepsilon^{\prime\prime})\frac{1}{\alpha-1}x\overline{F}(x).$$ Choose $q^{\prime\prime}>0$ such that $(1-\varepsilon)B^{\frac{1}{\alpha}} F^{\gets} \left ( q^{\prime\prime} \right )\ge x^{\prime\prime\prime}$. Then, for all $t\in\Lambda^{T}$ {and} all $q\ge \max\{q^{\prime},q^{\prime\prime}\}$, 
\begin{align}
\int\limits_{(1-\varepsilon)\left(\sum_{i=1}^{d}l_{i}(t)\right)^{\frac{1}{\alpha}} F^{\gets} \left ( q \right )}^{\infty}\overline{F}(z)\mathrm{d}z\le (1+\varepsilon^{\prime\prime})\frac{1-\varepsilon}{\alpha-1}\left(\sum_{i=1}^{d}l_{i}(t)\right)^{\frac{1}{\alpha}} F^{\gets} \left ( q \right )\overline{F}\left((1-\varepsilon)\left(\sum_{i=1}^{d}l_{i}(t)\right)^{\frac{1}{\alpha}} F^{\gets} \left ( q \right )\right).\label{lemma5.7pf2}
\end{align}
Also, by the regularly varying property of $F$ and Proposition 0.5 of \cite{resnick2008extreme}, it holds uniformly for all $t\in\Lambda^{T}$ that
\begin{align}
\overline{F}\left((1-\varepsilon)\left(\sum_{i=1}^{d}l_{i}(t)\right)^{\frac{1}{\alpha}} F^{\gets} \left ( q \right )\right)&\sim (1-\varepsilon)^{-\alpha}\left(\sum_{i=1}^{d}l_{i}(t)\right)^{-1}\overline{F}\left(F^{\gets} \left ( q \right )\right)\nonumber\\
&\sim (1-\varepsilon)^{-\alpha}\left(\sum_{i=1}^{d}l_{i}(t)\right)^{-1}(1-q).\label{lemma5.7pf3}    
\end{align}
Therefore, by (\ref{lemma5.7pf1}),  (\ref{lemma5.7pf2}) and (\ref{lemma5.7pf3}) and letting $\varepsilon, \varepsilon^{\prime}, \varepsilon^{\prime\prime}\downarrow 0$, we can get the upper-bound version of (\ref{lemma5.7.2}). The lower-bound version can be derived similarly.

To get (\ref{lemma5.7.3}), note that $$\E\left [ Z_{k,t}\id_{\left \{Z_{k,t}>F_{D_{t}}^{\gets }\left ( q \right ) \right \}} \right ]=F_{D_{t}}^{\gets }\left ( q \right ) \p\left ( Z_{k,t}>F_{D_{t}}^{\gets }\left ( q \right ) \right ) + \int\limits_{F_{D_{t}}^{\gets }\left ( q \right )}^{\infty}\p\left ( Z_{k,t}>z \right )\mathrm{d}z.$$ 
Thus, some asymptotic relations above can be appropriately applied to get the result.     
\end{proof}

\begin{lemma}\label{lemma7}
Under the {setting} of Theorem \ref{thm2}, for all $1\le k\le d$,  any fixed $\gamma\in (0,1),$ and any fixed $T\in\Lambda$, it holds uniformly for all $u\in [\gamma,1]$ and all $t\in\Lambda^{T}$ that
\begin{align}
\p\left(Z_{k,t}>uF_{D_{t}}^{\gets }(q),D_{t}>F_{D_{t}}^{\gets }(q)\right)\sim \p\left(Z_{k,t}>F_{D_{t}}^{\gets }(q)\right).\label{lemma5.8.1}
\end{align}    
\end{lemma}
\begin{proof}
 First notice that for all $u\in [\gamma,1]$ and all $t\in\Lambda^{T}$, 
\begin{align*}
\p\left(Z_{k,t}>uF_{D_{t}}^{\gets }(q),D_{t}>F_{D_{t}}^{\gets }(q)\right)\ge \p\left(Z_{k,t}>F_{D_{t}}^{\gets }(q),D_{t}>F_{D_{t}}^{\gets }(q)\right)=\p\left(Z_{k,t}>F_{D_{t}}^{\gets }(q)\right).
\end{align*}
On the other hand,
\begin{align}
\p\left(Z_{k,t}>uF_{D_{t}}^{\gets }(q),D_{t}>F_{D_{t}}^{\gets }(q)\right)&= \p\left(D_{t}>F_{D_{t}}^{\gets }(q)\right)-\p\left(Z_{k,t}\le uF_{D_{t}}^{\gets }(q),D_{t}>F_{D_{t}}^{\gets }(q)\right)\nonumber\\
&=:L_{1}(x,t)-L_{2}(x,t).\label{lemma5.8pf1}
\end{align}
Following the similar method of proving (\ref{lemma5.7.1}) in Lemma \ref{lemma6}, it is easy to see that uniformly for all $t\in\Lambda^{T}$,
\begin{align}
L_{1}(x,t)\sim 1-q\sim\sum_{p=1}^{d}\p\left(Z_{p,t}>F_{D_{t}}^{\gets }(q)\right).\label{lemma5.8pf2}
\end{align}
And as for $L_{2}(x,t)$, we have
\begin{align*}
L_{2}(x,t)&\ge \p\left(Z_{k,t}\le uF_{D_{t}}^{\gets }(q),\bigcup_{p=1,p\ne k}^{d}\{Z_{p,t}>F_{D_{t}}^{\gets }(q)\}\right)\\
&\ge\sum_{p=1,p\ne k}^{d}\p\left(Z_{k,t}\le uF_{D_{t}}^{\gets }(q),Z_{p,t}>F_{D_{t}}^{\gets }(q)\right)-\sum_{\substack{1\le p< w\le d,\\p\ne k,w\ne k}}\p\left(Z_{p,t}>F_{D_{t}}^{\gets }(q),Z_{w,t}>F_{D_{t}}^{\gets }(q)\right)\\
&=\sum_{p=1,p\ne k}^{d}\left(\p\left(Z_{p,t}>F_{D_{t}}^{\gets }(q)\right)-\p\left(Z_{k,t}> uF_{D_{t}}^{\gets }(q),Z_{p,t}>F_{D_{t}}^{\gets }(q)\right)\right)\\
&\quad-\sum_{\substack{1\le p< w\le d,\\p\ne k,w\ne k}}\p\left(Z_{p,t}>F_{D_{t}}^{\gets }(q),Z_{w,t}>F_{D_{t}}^{\gets }(q)\right).
\end{align*}  
So we only need to prove that, for all $1\le p\ne w\le d$, any fixed $\gamma\in (0,1)$ and all $T\in\Lambda$, it holds uniformly for all $u\in [\gamma,1]$ and $t\in\Lambda^{T}$ that $\p\left(Z_{p,t}>uF_{D_{t}}^{\gets }(q),Z_{w,t}>F_{D_{t}}^{\gets }(q)\right)=o(1)\p\left(Z_{p,t}>F_{D_{t}}^{\gets }(q)\right).$ We first consider the relation 
\begin{align}
\p\left(Z_{p,t}>x,Z_{w,t}>x\right)=o(1)\p\left(Z_{p,t}>x\right).\label{lemma5.8pf3}
\end{align}
For notational convenience, we set for all $m\in\mathbb{N}$ and all  $1\le k\le d$,
\begin{align*}
H_{k,t}:=\sum_{j=1}^{r}\sum_{i=1}^{m}X_{ki}^je^{-R_{\tau_{ki}^j}}\id_{\left \{ \tau_{ki}^j\le t \right \}} ,\:G_{k,t}:=\sum_{j=1}^{r}\sum_{i=m+1}^{\infty}X_{ki}^je^{-R_{\tau_{ki}^j}}\id_{\left \{ \tau_{ki}^j\le t \right \}}.  
\end{align*}
Notice that
\begin{align*}
&\quad\p\left(Z_{p,t}>x,Z_{w,t}>x\right)\\
&\le \p\left(\sum_{j=1}^{r}\sum_{i=1}^{\infty}X_{pi}^je^{-R_{\tau_{pi}^j}}1_{\left \{ \tau_{pi}^j\le t \right \}}>x,\,\sum_{j=1}^{r}\sum_{i=1}^{\infty}X_{wi}^je^{-R_{\tau_{wi}^j}}1_{\left \{ \tau_{wi}^j\le t \right \}}>x\right)\\
&\le \p\left(H_{p,t}>\frac{x}{2},H_{w,t}>\frac{x}{2}\right)+\p\left(G_{p,t}>\frac{x}{2},H_{w,t}>\frac{x}{2}\right)+\p\left(H_{p,t}>\frac{x}{2},G_{w,t}>\frac{x}{2}\right)\\
&\quad+\p\left(G_{p,t}>\frac{x}{2},G_{w,t}>\frac{x}{2}\right)\\
&=:L_{21}(x,t)+L_{22}(x,t)+L_{23}(x,t)+L_{24}(x,t).
\end{align*}
Now $L_{21}(x,t)$ can be ``split" into objects of the form $$\p\left( X_{ph}^ae^{-R_{\tau_{ph}^a}}\id_{\left \{ \tau_{ph}^a\le t \right \}}>\frac{x}{2rm}, X_{wl}^be^{-R_{\tau_{wl}^b}}\id_{\left \{ \tau_{wl}^b\le t \right \}}>\frac{x}{2rm} \right),$$ and then Lemma \ref{lemmabu} can be applied. $L_{22}(x,t)$, $L_{23}(x,t)$ and $L_{24}(x,t)$ can be handled through a similar method of proving Lemma \ref{lemma4}. Then, following a similar way in the proof of Theorem \ref{thm1}({\romannumeral1}), together with the fact that  $\p\left(Z_{p,t}>x\right)\asymp\p\left(Z_{w,t}>x\right)$ holds uniformly for all $t\in\Lambda^{T}$, we can get that (\ref{lemma5.8pf3}) holds uniformly for all $t\in\Lambda^{T}$.

Next, since $F_{D_{t}}^{\gets }\left ( q \right ) \sim\left(\sum_{i=1}^{d}l_{i}(t)\right)^{\frac{1}{\alpha}} F^{\gets} \left ( q \right )$ holds uniformly for all $t\in\Lambda^{T}$, applying a similar method as the proof of (\ref{lemma5.7.1}) in Lemma \ref{lemma6}, we can get for all $1\le p\ne w\le d$, uniformly for all $u\in [\gamma,1]$ and all $t\in\Lambda^{T}$, it holds that
\begin{align*}
\p\left(Z_{p,t}>uF_{D_{t}}^{\gets }(q),Z_{w,t}>F_{D_{t}}^{\gets }(q)\right)=o(1)\p\left(Z_{p,t}>uF_{D_{t}}^{\gets }(q)\right)=o(1)\p\left(Z_{p,t}>F_{D_{t}}^{\gets }(q)\right).
\end{align*}
Therefore, uniformly for any $u\in [\gamma,1]$ and $t\in\Lambda^{T}$, we have
\begin{align}
L_{2}(x,t)\gtrsim\sum_{p=1,p\ne k}^{d}\p\left(Z_{p,t}>F_{D_{t}}^{\gets }(q)\right).\label{lemma5.8pf4}
\end{align}
Combining (\ref{lemma5.8pf2}) and (\ref{lemma5.8pf4}), we get the upper-bound version of (\ref{lemma5.8.1}). Hence the result has been proven.   
\end{proof}

\begin{lemma}\label{lemma8}
Under the {setting} of Theorem \ref{thm2}, for any $1\le k\le d$ {and} any fixed $T\in\Lambda$, it holds uniformly for all $t\in\Lambda^{T}$ that
\begin{align*}
\E\left [ Z_{k,t}\id_{\left \{D_{t}>F_{D_{t}}^{\gets }\left ( q \right ) \right \}} \right ]\sim \E\left [ Z_{k,t}\id_{\left \{Z_{k,t}>F_{D_{t}}^{\gets }\left ( q \right ) \right \}} \right ].
\end{align*}    
\end{lemma}
\begin{proof}
First notice that $$\E\left [ Z_{k,t}\id_{\left \{D_{t}>F_{D_{t}}^{\gets }\left ( q \right ) \right \}} \right ]\ge \E\left [ Z_{k,t}\id_{\left \{Z_{k,t}>F_{D_{t}}^{\gets }\left ( q \right ) \right \}} \right ].$$ On the other hand, for any fixed $0<\beta<1$, 
\begin{align*}
\E\left [ Z_{k,t}\id_{\left \{D_{t}>F_{D_{t}}^{\gets }\left ( q \right ) \right \}} \right ]&= \E\left [ Z_{k,t}\id_{\left \{Z_{k,t}>\beta F_{D_{t}}^{\gets }\left ( q \right ),D_{t}>F_{D_{t}}^{\gets }\left ( q \right ) \right \}} \right ]+\E\left [ Z_{k,t}\id_{\left \{Z_{k,t}\le\beta F_{D_{t}}^{\gets }\left ( q \right ),D_{t}>F_{D_{t}}^{\gets }\left ( q \right ) \right \}} \right ]\\
&=:V_{1}(x,t)+V_{2}(x,t).
\end{align*}
Notice that
\begin{align*}
V_{1}(x,t)&=\beta F_{D_{t}}^{\gets }(q) \p\left(Z_{k,t}>\beta F_{D_{t}}^{\gets }(q),D_{t}>F_{D_{t}}^{\gets }(q)\right)+\int\limits_{\beta F_{D_{t}}^{\gets }\left ( q \right )}^{\infty}\p\left ( Z_{k,t}>z,D_{t}>F_{D_{t}}^{\gets }(q) \right )\mathrm{d}z\\
&=:V_{11}(x,t)+V_{12}(x,t).
\end{align*}
Applying Lemma \ref{lemma7}, we get that uniformly for all $t\in\Lambda^{T}$,
\begin{align}
V_{11}(x,t)&\sim\beta F_{D_{t}}^{\gets }(q) \p\left(Z_{k,t}>F_{D_{t}}^{\gets }(q)\right),\label{lemma5.9pf2}\\
V_{12}(x,t)&=\left(\int\limits_{\beta F_{D_{t}}^{\gets }\left ( q \right )}^{F_{D_{t}}^{\gets }\left ( q \right )}+\int\limits_{F_{D_{t}}^{\gets }\left ( q \right )}^{\infty}\right)\p\left ( Z_{k,t}>z,D_{t}>F_{D_{t}}^{\gets }(q) \right )\mathrm{d}z\nonumber\\
&=F_{D_{t}}^{\gets }\left ( q \right )\int\limits_{\beta}^{1}\p\left ( Z_{k,t}>uF_{D_{t}}^{\gets }\left ( q \right ),D_{t}>F_{D_{t}}^{\gets }(q) \right )\mathrm{d}u+\int\limits_{F_{D_{t}}^{\gets }\left ( q \right )}^{\infty}\p\left ( Z_{k,t}>z\right )\mathrm{d}z\nonumber\\
&\sim (1-\beta)F_{D_{t}}^{\gets }\left ( q \right )\p\left ( Z_{k,t}>F_{D_{t}}^{\gets }\left ( q \right )\right )+\int\limits_{F_{D_{t}}^{\gets }\left ( q \right )}^{\infty}\p\left ( Z_{k,t}>z\right )\mathrm{d}z.\label{lemma5.9pf3}
\end{align}
Combining (\ref{lemma5.9pf2}) and (\ref{lemma5.9pf3}), we get that uniformly for all $t\in\Lambda^{T}$,
\begin{align}
V_{1}(x,t)\sim F_{D_{t}}^{\gets }\left ( q \right )\p\left ( Z_{k,t}>F_{D_{t}}^{\gets }\left ( q \right )\right )+\int\limits_{F_{D_{t}}^{\gets }\left ( q \right )}^{\infty}\p\left ( Z_{k,t}>z\right )\mathrm{d}z=\E\left [ Z_{k,t}\id_{\left \{Z_{k,t}>F_{D_{t}}^{\gets }\left ( q \right ) \right \}} \right ].\label{lemma5.9pf4}
\end{align}
As for $V_{2}(x,t)$, notice that uniformly for all $t\in\Lambda^{T}$,
\begin{align}
V_{2}(x,t)&\le\sum_{p=1,p\ne k}^{d}\E\left [ Z_{k,t}\id_{\left \{Z_{k,t}\le\beta F_{D_{t}}^{\gets }\left ( q \right ),\,Z_{p,t}>\frac{1-\beta}{{d}-1}F_{D_{t}}^{\gets }\left ( q \right ) \right \}} \right ]\nonumber\\
&=\sum_{p=1,p\ne k}^{d}\int\limits_{0}^{\beta F_{D_{t}}^{\gets }\left ( q \right )}z\p\left ( Z_{k,t}\in \mathrm{d}z,\,Z_{p,t}>\frac{1-\beta}{{d}-1}F_{D_{t}}^{\gets }\left ( q \right )\right )\nonumber\\
&\le \sum_{p=1,p\ne k}^{d}\beta F_{D_{t}}^{\gets }\left ( q \right )\p\left (Z_{p,t}>\frac{1-\beta}{{d}-1}F_{D_{t}}^{\gets }\left ( q \right )\right )\nonumber\\
&\sim \beta\left(\frac{1-\beta}{{d}-1}\right)^{-\alpha}\sum_{p=1,p\ne k}^{d}\frac{l_{p}(t)}{l_{k}(t)}F_{D_{t}}^{\gets }\left ( q \right )\p\left (Z_{k,t}>F_{D_{t}}^{\gets }\left ( q \right )\right )\label{lemma5.9pf5}\\
&\le \beta\left(\frac{1-\beta}{{d}-1}\right)^{-\alpha}\sum_{p=1,p\ne k}^{d}\frac{l_{p}(t)}{l_{k}(t)}\E\left [ Z_{k,t}\id_{\left \{Z_{k,t}>F_{D_{t}}^{\gets }\left ( q \right ) \right \}} \right ],\nonumber
\end{align}
where the asymptotic relation \eqref{lemma5.9pf5} can be derived through a similar way as the proof of \eqref{lemma5.7.1} in Lemma \ref{lemma6}. Now due to the arbitrariness of $\beta$, we can make $\beta\left(\frac{1-\beta}{{d}-1}\right)^{-\alpha}$ as small as we desire. Therefore, uniformly for all $t\in\Lambda^{T}$, we have
\begin{align}
V_{2}(x,t)=o(1)\E\left [ Z_{k,t}\id_{\left \{Z_{k,t}>F_{D_{t}}^{\gets }\left ( q \right ) \right \}} \right ].\label{lemma5.9pf6} 
\end{align}
Combining \eqref{lemma5.9pf4} and \eqref{lemma5.9pf6}, we get that uniformly for all $t\in\Lambda^{T}$, $$\E\left [ Z_{k,t}\id_{\left \{D_{t}>F_{D_{t}}^{\gets }\left ( q \right ) \right \}} \right ]\lesssim \E\left [ Z_{k,t}\id_{\left \{Z_{k,t}>F_{D_{t}}^{\gets }\left ( q \right ) \right \}} \right ].$$ Then we get the result.    
\end{proof}

Now, we can turn to the proof of Theorem \ref{thm2}.\\
\textbf{Proof of Theorem \ref{thm2}.} First notice that 
\begin{align*}
\SES_{q,k}(D_{t})&=\int\limits_{F_{Z_{k,t}}^{\gets }(q)}^{\infty}\p\left(Z_{k,t}>z\mid D_{t}>F_{D_{t}}^{\gets }(q)\right)\mathrm{d}z\\
&=\frac{1}{\p\left(D_{t}>F_{D_{t}}^{\gets }(q)\right)}\int\limits_{F_{Z_{k,t}}^{\gets }(q)}^{\infty}\p(Z_{k,t}>z, D_{t}>F_{D_{t}}^{\gets }(q))\mathrm{d}z\\
&=\frac{F_{D_{t}}^{\gets }(q)}{\p\left(D_{t}>F_{D_{t}}^{\gets }(q)\right)}\left(\int\limits_{\frac{F_{Z_{k,t}}^{\gets }(q)}{F_{D_{t}}^{\gets }(q)}}^{1}+\int\limits_{1}^{\infty}\right)\p\left(Z_{k,t}>uF_{D_{t}}^{\gets }(q), D_{t}>F_{D_{t}}^{\gets }(q)\right)\mathrm{d}u\\
&=:\frac{F_{D_{t}}^{\gets }(q)}{\p\left(D_{t}>F_{D_{t}}^{\gets }(q)\right)}(A_{1}(x,t)+A_{2}(x,t)).
\end{align*}
It is easy to see 
\begin{align}
A_{2}(x,t)=\int\limits_{1}^{\infty}\p\left(Z_{k,t}>uF_{D_{t}}^{\gets }(q)\right)\mathrm{d}u=\frac{1}{F_{D_{t}}^{\gets }(q)}\E\left [ \left ( Z_{k,t}-F_{D_{t}}^{\gets }\left ( q \right ) \right )^{+} \right ].\label{thm2pf1}
\end{align}
As for $A_{1}(x,t)$, for all $t\in\Lambda^{T}$, $\lim_{q\uparrow 1}\frac{F_{Z_{k,t}}^{\gets }(q)}{F_{D_{t}}^{\gets }(q)}=\left(\frac{l_{k}(t)}{\sum_{i=1}^{d}l_{i}(t)}\right)^{\frac{1}{\alpha}}<1$ (the limit of $F_{Z_{k,t}}^{\gets }(q)$ can be derived through a similar way as the proof of Lemma \ref{lemma5}). By Lemma \ref{lemma7}, we can know that uniformly for all $t\in\Lambda^{T}$, 
\begin{align}
A_{1}(x,t)\sim \left(1-\frac{F_{Z_{k,t}}^{\gets }(q)}{F_{D_{t}}^{\gets }(q)}\right)\p\left(Z_{k,t}>F_{D_{t}}^{\gets }(q)\right).\label{thm2pf2}
\end{align}
Combining \eqref{thm2pf1}, \eqref{thm2pf2} and Lemma \ref{lemma8}, we get that uniformly for all $t\in\Lambda^{T}$,
\begin{align*}
&\SES_{q,k}(D_{t}) \sim \frac{\p\left ( Z_{k,t}> F_{D_{t}}^{\gets }\left ( q \right ) \right ) \left(F_{D_{t}}^{\gets }(q)-F_{Z_{k,t}}^{\gets }(q)\right)
+\E\left [ \left ( Z_{k,t}-F_{D_{t}}^{\gets }\left ( q \right ) \right )^{+}  \right ] }{\p\left (D_{t}>F_{D_{t}}^{\gets }\left ( q \right ) \right ) }, \\
&\MES_{q,k}(D_{t})=\frac{\E\left [ Z_{k,t}\id_{\left \{D_{t}>F_{D_{t}}^{\gets }\left ( q \right ) \right \}} \right ]}{\p\left ( D_{t}>F_{D_{t}}^{\gets }\left ( q \right ) \right )} \sim \frac{\E\left [ Z_{k,t}\id_{\left \{Z_{k,t}>F_{D_{t}}^{\gets }\left ( q \right ) \right \}} \right ]}{\p\left ( D_{t}>F_{D_{t}}^{\gets }\left ( q \right ) \right )}, 
\end{align*}
Now applying \eqref{uniform1}, \eqref{lemma5.7.1}, \eqref{lemma5.7.2}, \eqref{lemma5.7.3} and \eqref{lemma5.8pf2}, we can get that uniformly for all $t\in\Lambda^{T}$,
\begin{align*}
&\SES_{q,k}(D_{t}) \sim \frac{l_{k}(t)}{\sum_{i=1}^{d}l_{i}(t)} \left ( \left(\sum_{i=1}^{d}l_{i}(t)\right)^{\frac{1}{\alpha}}-\left ( l_{k}(t) \right )^{\frac{1}{\alpha}} + \frac{\left ( \sum_{i=1}^{d}\left (l_{i}(t)  \right ) \right )^{\frac{1}{\alpha} } }{\alpha-1}\right ) F^{\gets }\left ( q \right),\\
&\MES_{q,k}(D_{t}) \sim \frac{\alpha }{\alpha-1} \frac{l_{k}(t)}{\left( \sum_{i=1}^{d}l_{i}(t) \right )^{1-\frac{1}{\alpha}}} F^{\gets }\left ( q \right ). 
\end{align*} 
Thus, the proof is completed. $\blacksquare$

\subsection{On Proposition 4.1}
By simple calculations, we get 
\begin{align*}
p^A(x,t)=&\sum_{k=1}^d\sum_{j=1}^{r}\overline{F_{k}^j}(x)\lambda_k^j\int\limits_0^te^{s\phi_R(\alpha)}\mathrm{d}s;\\
\mathrm{SES}^A(q,k,t)(D_{t})=&\frac{\sum_{j=1}^{r}a_k^j\lambda_k^j}{\sum_{i=1}^d\sum_{j=1}^{r}a_i^j\lambda_i^j}\left[\frac{\alpha}{\alpha-1}\left(\sum_{i=1}^d\sum_{j=1}^{r}a_i^j\lambda_i^j\right)^{\frac{1}{\alpha}}-\left(\sum_{j=1}^{r}a_k^j\lambda_k^j\right)^{\frac{1}{\alpha}}\right]\\
&F^{\gets}(q)\left(\int\limits_0^te^{s\phi_R(\alpha)}\mathrm{d}s\right)^{\frac{1}{\alpha}};\\
\mathrm{MES}^A(q,k,t)(D_{t})=&\frac{\sum_{j=1}^{r}a_k^j\lambda_k^j}{\sum_{i=1}^d\sum_{j=1}^{r}a_i^j\lambda_i^j}\frac{\alpha}{\alpha-1}F^{\gets}(q)\left(\int\limits_0^te^{s\phi_R(\alpha)}\mathrm{d}s\right)^{\frac{1}{\alpha}}.
\end{align*}
Then, notice that in the three equations above, all components except $\int\limits_0^te^{s\phi_R(\alpha)}\mathrm{d}s$ are positive constants. Thus, we only need to consider the impact of $\phi_R(\alpha)$. From (\ref{phi_Rt}), it follows that the function $\phi_R(\alpha)$ decreases with $\mu$ and increases with $\sigma$ and $\lambda$. Then the proposition is proven. \\

\noindent
{\bf Acknowledgements.}
{The authors are grateful to the editors and two anonymous referees for valuable comments that have greatly improved this paper.}
B. Geng acknowledges financial support from the research startup fund (Grant No. S020318033/015)
at Anhui University and the Provincial Natural Science Research Project of Anhui Colleges (Grant No. 2024AH050037). Y. Liu acknowledges financial support from the National Natural Science Foundation of China (Grant No. 12401624), The Chinese University of Hong Kong, Shenzhen research startup fund (Grant No. UDF01003336) and Shenzhen Science and Technology Program (Grant Nos. RCBS20231211090814028, JCYJ20250604141203005, 2023TC0177), and is partly supported by the Guangdong Provincial Key Laboratory of Mathematical Foundations for Artificial Intelligence (Grant No. 2023B1212010001).  
The authors are grateful to members of the research group on financial mathematics and risk management at The Chinese University of Hong Kong, Shenzhen for their useful feedback and conversations.


\small
\bibliographystyle{apalike}
\bibliography{arxiv}

\end{document}